\documentclass[10pt]{article} %
\usepackage[english]{babel}
\usepackage{a4}
\usepackage{latexsym}
\usepackage{amssymb}

\newenvironment{tab}{\begin{tabbing}
MMMMM\=aaa\=aaa\=aaa\=aaa\=aaa\=aaa\= \kill}{\end{tabbing}}

\makeatletter
\def\refmystepcounter#1{\stepcounter{#1}\protect\gdef 
\@currentlabel {\csname p@#1\endcsname \csname 
the#1\endcsname}}
\makeatother
\newcounter {tabnr}
\newenvironment{tabn}{\begin{tabbing}
\refmystepcounter{tabnr}
MMMMM\=aaa\=aaa\=aaa\=aaa\=aaa\=aaa\= \kill
(\arabic{tabnr})~}{\end{tabbing}}

\newtheorem{theorem}{Theorem}

\newenvironment{proof}{\noindent\emph{Proof.}}{\boks}
\newenvironment{remark}{\medbreak\noindent\emph{Remark.}}{\boks}

\def\sem #1{\hbox{$[\![\, #1\, ]\!] $}}
\def\boks  {\mbox{$\Box$}}
\def\Nat   {\mbox{$\mathbb{N}$}}

\def\bar   {\mbox{$\,[ \! ]\,$}}

\def\all   {\forall\;}
\def\ex    {\exists\;}

\def\S #1/{\mbox {\textsl{#1}}}
\def\B #1/{\mbox {\textbf{#1}}}
\def\R #1/{\mbox {\textrm{#1}}}
\def\T #1/{\mbox {\texttt{#1}}}

\def\phi   {{\mbox{$\varphi$}}}

\def\Implies{\;\Rightarrow\;}

\def\EQ     {\mbox{\quad$\equiv$\quad}}

\def\Land   {\mbox{ $\;\land\;$ }}

\def\Lor    {\;\lor\;}

\def\TO     {\mbox{$\quad\to\quad$}}
\def\true   {\S true/}
\def\false  {\S false/}
\def\IS     {\mbox{$\quad =\quad $}}

\def\sbreak {\smallbreak\noindent}

\def\bbreak {\bigbreak\noindent}

\def\bol    {\mbox{$\bullet$}}
\begin{document}

\title {A distributed resource allocation algorithm\\
    for many processes}
\author {Wim H. Hesselink \ (whh469)\\
  University of Groningen, The Netherlands\\
  w.h.hesselink@rug.nl
}
\maketitle

\setcounter{tabnr}{-1}

\begin{abstract}
\noindent
Resource allocation is the problem that a process may enter a critical
section \S CS/ of its code only when its resource requirements are not
in conflict with those of other processes in their critical sections.
For each execution of \S CS/, these requirements are given anew. In
the resource requirements, levels can be distinguished, such as
e.g. read access or write access.  We allow infinitely many processes
that communicate by reliable asynchronous messages and have finite
memory.  A simple starvation-free solution is presented.  Processes
only wait for one another when they have conflicting resource
requirements.  The correctness of the solution is argued with
invariants and temporal logic. It has been verified with the proof
assistant PVS.
\end{abstract}

\bbreak \B Key words:/ distributed algorithms; resource allocation;
drinking philosophers; readers/writers problem; verification;
starvation freedom; fairness

\section {Introduction}

Resource allocation is a problem that goes back to Dijkstra's dining
philosophers \cite{Dij71} and the drinking philosophers of Chandi and
Misra \cite{ChM84}. It is the problem that a process may enter a
critical section of its code only when its resource requirements are
not in conflict with those of other processes in their critical
sections. In the case of the dining philosophers, the philosophers
form a ring and the resource requirements are two forks shared with
the neighbours in the ring.  In the drinking philosophers problem the
philosophers form an arbitrary finite undirected graph.

In the general resource allocation problem, there is a number of
processes that from time to time need to execute a critical section \S
CS/ in which they need access to some resources.  For every critical
section of every process, the resource requirements may be different.
Processes that are concurrently in a critical section, must have
compatible resource requirements.  The processes must therefore
communicate with possible competitors, and possibly wait for
conflicting processes before entering \S CS/. On the other hand,
unnecessary waiting must be avoided.

\subsection{Setting and sketch of solution}

We present a solution for a setting with infinitely many processes
that have private memory and communicate by asynchronous messages.
The processes and messages are assumed to be reliable: the processes
never crash, the mesages are guaranteed to arrive and be handled, but
the delay is unknown.  Messages are not lost, damaged, or duplicated.
They can pass each other, however, unlike in \cite{AwS90,Lyn96} where
the messages in transit from one sender to one receiver are treated
first-in-first-out.  Every process can send messages to every other
process, and receive messages from it.  The processes receive and
answer messages even when they are idle.

The resource requirements can be sets of resources the processes need
exclusive access to.  More generally, however, resource requirements
may comprise, e.g., read access or write access to some data, where
concurrent reading is allowed while writing requires exclusive access.

In our solution, we deal with the infinitely many processes by
splitting the problem in two parts: one part to ensure that a process
has, for every job, only a finite set of potential competitors, its
dynamic neighbourhood, the other part to use this neighbourhood to
ensure partial mutual exclusion.

The two parts interact mildly.  We call the first part the
\emph{registration algorithm}, because it is based on the idea that
processes need to register for resources.  The second part is called
the \emph{central algorithm} because it satisfies the functional
requirements.  If the neighbourhoods can be kept constant, the
registration algorithm can be removed, and the central algorithm can
be compared with the drinking philosophers \cite{ChM84}.

When a process gets a new job, it first registers to obtain a list of
potential competitors.  If in doing this it extends its registrations,
it may need to contact some of these competitors before proceeding.
Then, the central algorithm takes over.  This first lets the process
wait for conflicting processes that are currently competing for \S
CS/, and then allows it in the competition for \S CS/.  A process with
many registrations has much communication to perform.  We therefore
offer the processes the option to withdraw registrations, in some
degree concurrently with resource acquisition.

The algorithm has two kinds of waiting conditions: waiting for
messages in transit to arrive, and waiting for conflicting processes
to proceed.  It is not our aim to minimize the waiting time.  We offer
a simple solution with as much nondeterminism as possible and as much
progress as we can accommodate in view of the safety requirements.

\subsection{Overview and verification}

We briefly discuss related research in Section \ref{relatedresearch}.
In Section \ref{async}, we describe the model and the notations for
message passing. 

Section \ref{intromain} presents the algorithms.  Section
\ref{algorithm} contains its proof of safety.  In Section
\ref{nodead}, we introduce and formalize progress, and prove that weak
fairness guarantees progress for the registration algorithm.  In
Section \ref{proofThm1}, we define, formalize, and prove progress of
the central algorithm in a form that combines starvation freedom and
concurrency.  We discuss message complexity and waiting times in
Section \ref{summ_alg}, and conclude in Section \ref {conclusion}.

The proofs of the safety and liveness properties have been carried out
with the interactive proof assistant PVS \cite{OSR01}.  The
descriptions of proofs closely follow our PVS proof scripts, which can
be found on our web site \cite{whh_distrRscAlloc}.  It is our
intention that the paper can be read independently, but the proofs
require so many case distinctions that manual verification is
problematic.

\subsection{Related research} \label{relatedresearch}

The readers/writers problem \cite{And00} goes back to Courtois,
Heymans, and Parnas \cite{CHP71}, in the context of shared memory
systems with semaphores.  We are not aware of solutions for systems
with message passing.

In the drinking philosophers' problems of \cite{ChM84,Lyn96,WeL93},
the philosophers form a fixed finite undirected graph.  In this case,
the set $\S nbh/.p$ of possibly conflicting processes of a process $p$
is a subset of the constant set $\S Nbh/.p$ of $p$'s neighbours in the
graph.  This subset is chosen nondeterministically when the process
becomes ``thirsty''. The message complexity of the solutions in these
papers is proportional to the size of $\S Nbh/.p$ and not to the
possibly considerably smaller size of $\S nbh/.p$.  This is a
disadvantage for cases with large complete graphs. To enforce
starvation freedom, these solutions assign directions to the edges of
the graph such that the resulting directed graph is acyclic.

Much work has been done to minimize the response time
\cite{AwS90,ChS92,PJC93,Rhe98,WeL93}. For instance, the paper
\cite{WeL93} offers the possibility of a waiting time that is constant
and not proportional to the (in our case unbounded) number of
processes. It does so by means of the algorithm of \cite{Lyn81}, which
uses a linear ordering of the resources adapted to a fixed netwerk
topology.  The papers \cite{Rhe98,WeL93} offer modular approaches to
the general resource allocation problem.

Another important performance aspect is robustness against
failures. The paper \cite{ChS92} introduces the measure of failure
locality, see also \cite{SPS00}. The paper \cite{DGR05} concentrates
on self-stabilization, while imposing specific conditions on the
resource requirements.

As far as we can see the only papers that treat the dynamic resource
allocation problem that allows conflicts between arbitrary pairs of
processes are \cite{AwS90,Rhe98}.  The algorithm of \cite{AwS90}
ignores the resources and takes the conflicts as given. Whenever two
processes have conflicting jobs, at least one of the two is activated
with knowledge of the conflict.  The emphasis of \cite{AwS90,Rhe98} is
on minimizing the response time.  The algorithms are more complicated
and need more messages than our solution.  The paper \cite{Rhe98} uses
resource managers.

\subsection{Asynchronous messages} \label{async}

The processes communicate by reliable asynchronous messages, i.e., the
messages are not corrupted, lost, or duplicated.  They may, however,
pass each other.

Every message has a message key, a sender, and a unique destination.
It may have a value.  As in CSP \cite{Hoa85}, we write $m.q.r\,!$ for
the command for $q$ to send a message with key $m$ to destination $r$,
and $m.q.r\,?$ for the command for $r$ to receive this message.
Unlike CSP, the messages are asynchronous.  We write $m.q.r\,!\,v$ for
the command to send a message with key $m$ and value $v$ from $q$ to
$r$, and $m.q.r\,?\,v$ for the command to receive message $m$ and
assign its value to $v$.

In the algorithm, for every message key $m$, and for every source $q$
and destination $r$, there is never more than one message in transit
from $q$ to $r$. Therefore, e.g., in Promela, the language of the
model checker Spin \cite{Hol04}, one could model the messages by
channels with buffer size 1.  

For the correctness of the algorithm, the time needed for message
transfer can be unbounded. For other issues, however, it is convenient
to postulate an upper bound $\Delta$ for the time needed to execute an
atomic command plus the time that the messages sent in this command
are in transit.  Similarly, when discussing progress, we assume that
the execution time of the critical sections is bounded by $\Gamma$.

\section {The  Algorithm} \label{intromain}
\label{fifoalg}

Section \ref{funcspec} contains the functional specification of the
complete algorithm, and separates the responsibilities of the central
algorithm and the registration algorithm.   Progress requirements are
discussed in Section \ref {progress}.  

The central algorithm is sketched in Section \ref{sketch}. Section
\ref{code} contains the code of the central algorithm and discusses
some global aspects.  In Section \ref{layered}, we present its design
as a layered algorithm. 

In Section \ref{jobmodel}, we develop a job model as a preparation for
the registration algorithm presented in Section \ref{queryalg}.
Section \ref{s.abort} contains commands to abort the entry protocol.

The entire algorithm is presented in this section without
verification.  The verification is postponed to Section
\ref{algorithm}.  Yet, we designed the algorithm concurrently with the
verification, because that is for us the only way to obtain a reliable
algorithm.  We separate the two aspects here for the ease of reading.

\subsection{Functional specification} \label{funcspec}\label{safetyreq}

When a process gets a new job to execute in a critical section, the
associated resource requirements are chosen nondeterministically, say
by a (distributed) environment.  We use the term \emph{job} for these
resource requirements, and we use the type \S Job/ for all possible
jobs.  The value $\S none/:\S Job/$ represents the absence of resource
requirements.  We give every process $q$ a private variable $\S job/.q
:\S Job/$, which is initially \S none/, but which is modified by the
environment when it gives the process a job.

Processes concurrently in the critical section \S CS/ must have
\emph{compatible} jobs.  We write $u*v$ to express that the jobs $u$
and $v$ are compatible.  We assume that compatibility satisfies the
axioms that $u*\S none/\equiv \true$ and $u*v\equiv v*u$ for all jobs
$u$ and $v$.  Examples of compatibility relations are given in Section
\ref{jobmodel}.

The problem of \emph{resource allocation} is thus to ensure that
processes with $\S job/\ne\S none/$ eventually enter \S CS/, under
the safety requirement that, when two different processes are both in
\S CS/, their jobs are compatible:
\begin{tab}
\S Rq0:/ \> $q\B\ in /\S CS/\Land r\B\ in /\S CS/ \Implies q = r \Lor
\S job/.q * \S job/.r $ .
\end{tab}
Here and henceforth, $q$ and $r$ stand for processes.  For all
invariants, we implicitly universally quantify over the free
variables, usually $q$ and $r$.  If \S v/ is a private variable,
outside the code, the value of \S v/ for process $q$ is denoted by $\S
v/.q$.

We speak of a \emph{conflict} between $q$ and $r$ when $q$ and $r$
have incompatible jobs: $\S job/.q * \S job/.r \equiv \false $.
Clearly, the conflict relation is time dependent.  Condition \S Rq0/
says that conflicting processes are never concurrently in their
critical sections.  For comparison, mutual exclusion itself would be
the requirement that $ q\B\ in /\S CS/$ and $r\B\ in /\S CS/$ implies
$ q = r $.

While we allow potentially infinitely many processes, we cannot expect
a process to communicate with infinitely many competitors.  We
therefore introduce a \emph{registration algorithm} that uses the new
job of a process $q$ to provide $q$ with a finite set $\S nbh0/.q$ of
potential competitors.  Every process not in $\S nbh0/.q$ must not in
conflict with $q$.  In other words, the registration algorithm serves
to guarantee the invariant:
\begin{tab}
\S Rq1:/ \> $ q\B\ in /\S CS/ \Land r\B\ in /\S CS/ \Implies
q = r \Lor r\in\S nbh0/.q \Lor \S job/.q * \S job/.r $ .
\end{tab}
The \emph{central algorithm} uses the sets $\S nbh0/.p$ to guarantee
the invariant:
\begin{tab}
\S Rq2:/ \> $ q\B\ in /\S CS/\Land r\B\ in /\S CS/ \Land 
r\in\S nbh0/.q\Land q\in\S nbh0/.r \Implies 
\S job/.q * \S job/.r $ .
\end{tab}
Predicate \S Rq0/ follows from \S Rq1/ and \S Rq2/, using the symmetry
of the operation $*$.  The reason for the name \S nbh0/ is that the
algorithm also uses a closely related variable \S nbh/ which does not
satisfy \S Rq2/.

As a process may have to wait a long time before it can enter \S CS/,
it may be useful that the environment of a process can
nondeterministically abort its entry protocol and move it back to the
idle state.  This option is offered in Section \ref{s.abort}.

\subsection{Progress requirements} \label{progress}

The first progress requirement that comes to mind, is \emph{starvation
  freedom} \cite{ewd651}, also called lockout freedom \cite{Lyn96}.
This means that every process that needs to enter \S
CS/, will eventually do so unless its entry protocol is aborted.

While starvation freedom is important, resource allocation has a
second requirement, viz.\ that no process is hindered unnecessarily.
This property is called \emph{concurrency} in \cite{ChM84,Rhe98}.  It
means that every process that needs to enter \S CS/ and does not
abort, will eventually enter \S CS/, unless it comes in eternal
conflict with some other process (a kind of deadlock).

Of course, concurrency follows from starvation freedom.  We introduce
both, however, because concurrency needs a weaker liveness assumption
than starvation freedom.  The liveness assumptions needed are forms of
weak fairness.  Weak fairness for process $p$ means that if, from some
time onward, process $p$ is continuously enabled to do a step
different from aborting, it will do the step. Weak fairness for
messages $m$ from $q$ to $r$ means that every message $m$ in transit
from $q$ to $r$ will eventually arrive.  Weak fairness is a natural
assumption, and some form of it is clearly needed.  We come back to
this in Section \ref{intro-wf}.

Our algorithm satisfies starvation freedom under the assumption of
weak fairness for all processes and all messages.  Concurrency for
process $p$ only needs weak fairness for $p$ itself and weak fairness
for all messages from and to $p$.  The point is that progress of
process $p$ is not hindered by unfair processes without conflicts with
$p$; e.g., such processes are allowed to remain in \S CS/ forever.  In
Section \ref{s.liveness}, we unify starvation freedom and concurrency
in a single progress property called absence of localized starvation.

\subsection{Sketch of the central algorithm} \label{sketch}

The central algorithm thus works under the assumption that, for every
\S CS/ execution, the process obtains a finite set \S nbh/ of
potential competitors with the guarantee that the \S CS/ execution
does not lead to a conflict with other processes.  The elements of \S
nbh/ are called the neighbours of the process.  We do not yet
distinguish \S nbh/ and \S nbh0/.

Inspired by the shared-variable mutual exclusion algorithm of
\cite{LyH91}, the central algorithm is designed in three layers: an
outer protocol to communicate the job to all neighbours, a middle
layer to guarantee starvation freedom and to guard against known
conflicts, and an inner protocol to guarantee mutual exclusion.

The inner protocol is the competition for \S CS/.  It uses process
numbers for tie breaking, just as Lamport's Bakery algorithm
\cite{Lam74}. We therefore represent the processes by natural
numbers. If $q$ and $r$ are processes with $q<r$, we speak of $q$ as
the lower process and $r$ as the higher process.  For every pair of
processes, the inner protocol gives priority to the lower process, and
it lets the higher process determine compatibility.  As processes can
enter the inner protocol concurrently only within the margins allowed
by the middle layer, we expect that the priority bias of the inner
protocol is not very noticeable unless the load is so heavy that the
performance of any algorithm would be problematic.

Despite the three layers, the central algorithm is rather simple.  The
outer protocol uses three messages for every process in \S nbh/.  The
middle layer needs no additional messages.  The inner protocol uses
one message for every higher process in \S nbh/, and no messages for
the lower processes in \S nbh/.

\subsection{Into the code of the central algorithm} \label{code}

The code of the central algorithm is given in Figure \ref{fifo}.  For
an unbroken flow of control, we include the lines 22 and 23, which
belong to the registration algorithm.

Every process $p$ has a program counter $\S pc/.p:\Nat$, which is
initially 21.  For a process $p$ and a line number $\ell$, we write
$p\B\ at / \ell$ to express $\S pc/.p = \ell$. If $L$ is a set of line
numbers, we write $p\B\ in / L$ to express $\S pc/.p \in L$.

Every process, say $p$, is in an infinite loop, in the line numbers 21
up to 28. It is at line 21 if and only if it is idle. Independently of
its line number, it is always able and willing to receive messages
from any other process, say $q$.  In other words, in Figure
\ref{fifo}, the eight alternatives of \T central/ are interleaved with
the six alternatives of \B receive/.

We regard the 14 alternatives of Figure \ref{fifo} as atomic commands.
This is allowed because actions on private variables give no
interference, the messages are asynchronous, and any delay in sending
a message can be regarded as a delay in message delivery.

One may regard the receiver as an independent thread of the looping
process, with the guarantee that the 14 alternatives are executed
atomically.  Alternatively, an implementer may decide to schedule the
receiver only when the process is idle (at line 21) or waiting at
lines 22, 23, 24, 25, 26.

\begin{figure}[t]
\begin{tab}
  \> $ \B central/(p): $ \+\\
  \B loop/\\
  21: \>\> $ \B choose /\S job/\ne\S none/ $ ;\\
  22: \>\> $ \B await /\S pcr/ \leq 32 $ ;\\
  \>\> $ \S curlist/:= \{ s \mid L(\S job/)(s) > 0\} $ ;\\
  \>\> $ \B for all /s\in\S curlist/\B\ do /
  \T asklist/.p.s\:!\:L(\S job/)(s) \B\ od/ $ ;\\
  23: \>\> $ \B await /\S curlist/ = \emptyset $ ;
  // \S nbh/ has been formed.\\
  \>\> $ \B for all /q\in\S pack/ \B\ do /\T hello/.p.q\:!\B\ od/ $ ;\\
  24: \>\> $ \B await /\S pack/ \cup \S wack/ = \emptyset $ ;\\
  \>\> $ \S prio/ := \{q \mid q\notin \S after/ \Land
  \neg\, \S job/*\S copy/(q) \} $ ;\\
  25: \>\> $ \B await / \S prio/ = \emptyset $ ; $ \S nbh0/ := \S nbh/ $ ;\\
  \>\> $ \B for all /q\in\S nbh/ \B\ do /\T notify/.p.q\:!\:\S job/\B\ od/ $ ;\\
  \>\> $ \S need/:= \{q\in\S nbh/\mid p < q \lor (q\in \S away/
  \land \neg\,\S job/*\S copy/(q)) \} $ ;\\
  26: \>\> $ \B await / \S need/ = \emptyset $ ;\\
  27: \>\> $ \S CS/ $ ;\\
  28: \>\> $ \B for all /q\in\S nbh/ \B\ do /\T withdraw/.p.q\:!\B\ od/ $ ;\\
  \>\> $ \S wack/ := \S nbh/ $ ;
  $ \S job/ := \S none/ $ ; $\S nbh/ := \S nbh0/ := \emptyset $ ;\\
  \B endloop/ .
  \\
  \\
  $ \B receive/(p) \B\ from/(q): $\+ \\

  $ \bar $\> $  \T notify/.q.p \:?\:\S copy/(q) $ ;\\
  \> $ \B if /q < p \B\ then / \B add /q \B\ to / \S prom/\B\ endif/ $ .\\
  
  $ \bar $\> $  \T withdraw/.q.p\: ? $ ;\\
  \> $\B add /q\B\ to /\S after/ $ ;  $ \B remove /q\B\ from /\S prio/ $ ; \\
  \> $ \B if / q < p \B\ then /
  \B remove /q\B\ from /\S away/ \B\ and /\S need/ \B\ endif/ $ .\\

  $ \bar $ \> \T after:/ \quad $ q\in\S after/\Land \S copy/(q)\ne\S none/\TO$\\
  \> $ \T ack/.p.q \:! $ ; 
  $ \B remove /q\B\ from /\S after/$ ; $ \S copy/(q):=\S none/ $ . \\

  $ \bar $\> $  \T ack/.q.p \:? $ ;
  $ \B remove /q\B\ from /\S wack/ $ . \\

  $ \bar $\> $  \T gra/.q.p\: ? $ ; $ \B remove /q \B\ from /\S need/ $ .\\
 
  $\bar$ \> \T prom:/ \quad $ q\in\S prom/\Land (\S pc/ \leq 26
  \Lor \S job/*\S copy/(q)) \TO $ \\
  \> $ \T gra/.p.q\: ! $ ; $ \B add /q\B\ to /\S away/ $ ;
  $ \B remove /q\B\ from /\S prom/ $ ;\\
  \> $ \B if / \S pc/ = 26 \Land  \neg\, \S job/*\S copy/(q) $\\
  \> $ \B then /\B add /q\B\ to /\S need/\B\ endif/ $ .\-\\
  \B end receive/ .
\end{tab}
\caption{The central algorithm for process $p$ (with $p$'s private
  variables)}\label{fifo}
\end{figure}

Every process has a private variable \S job/ of type \S Job/,
initially $\S none/$. It has private variables \S nbh/, \S nbh0/, \S
prio/, \S wack/, \S after/, \S away/, \S need/, \S prom/, which all
hold finite sets of processes.  All these sets are initially empty.
It has the private variables \S pcr/, \S pack/, \S fun/, and \S
curlist/, which serve in the registration algorithm and are treated in
Section \ref{queryalg}.  The private variable \S nbh0/ is a history
variable.  It is set to \S nbh/ when the process executes line 25 and
reset when the process leaves 28. It is not used in the algorithm, but
serves in the proof of correctness.

Process $p$ has a private extendable array 
$\S copy/.p$, such that $\S copy/.p(q)=\S job/.q$ holds under suitable
conditions.  It is set when receiving $\T notify/.q.p$ and reset in \T
after/.  We use the convention that $\S copy/.p(q) = \S none/ $ when
$q$ is not in the current range of the array. Initially, the range of
$\S copy/.p$ is empty.

At line 21, the environment gives process $p$ a job.  The lines 22 and
23 are treated in Section \ref{queryalg} with the registration
algorithm.

For now we just assume that $\S nbh/.p$ gets some value before $\S
curlist/.p$ becomes empty at line 23 and that, somehow, predicate
\S Rq1/ is guaranteed.  

The central algorithm uses four message keys: \T notify/, \T
withdraw/, \T ack/, \T gra/. The messages \T notify/ hold values of
the type \S Job/, the other messages hold no values, they are of type
\T void/.  The alternatives of \B receive/ with labels \T after/ and
\T prom/ correspond to delayed answers.  The message key \T asklist/
is treated in Section \ref{queryalg}.

\subsection{A layered solution} \label{layered}

As announced, the central algorithm has three layers: an outer
protocol to communicate the jobs, a middle layer to regulate access to
the inner protocol, and an inner protocol for mutual exclusion.  The
three layers have waiting conditions in the lines 24, 25, 26,
respectively.  The outer protocol uses the messages \T notify/, \T
withdraw/, \T ack/, and the private variables \S job/, \S nbh/, \S
wack/, and \S copy/. It can be obtained from Fig. \ref{fifo} by
removing the lines 22, 23, and all commands that use the messages \T
gra/ and the private variables \S prio/, \S need/, \S away/, \S prom/.
The middle layer consists of all commands that use \S prio/.  The
inner protocol consists of the commands that use the messages \T gra/
and the private variables \S need/, \S away/, \S prom/.

\subsubsection{The outer protocol } \label{introOuter}

In the outer protocol, every process $p$ sends its \S job/ to all
neighbours by means of \T notify/ messages in line 25, and it
withdraws this in line 28. Reception of \T notify/ and \T withdraw/ is
handled in the first three alternatives of \B receive/.  In the first
line of \T notify/, process $p$ registers the \S job/ of $q$ in $\S
copy/.p(q)$.  The conditional statement of \T notify/ belongs to the
inner protocol.  When process $p$ has received both \T notify/ and \T
withdraw/ from $q$, it can execute the alternative \T after/ of \B
receive/, send \T ack/ back to $q$, and reset $\S copy/.p(q):= \S
none/$.  In this way, we allow the message \T withdraw/ to arrive
before \T notify/, even though it was sent later.

In the fourth alternative of \B receive/, when process $p$ receives an
\T ack/ from $q$, it removes $q$ from its set \S wack/. This variable
has been set by $p$ to $\S nbh/.p$ in line 28, while sending \T
withdraw/ to its neighbours.  When process $p$ arrives again at line
24, it waits for \S wack/ to be empty. In this way, it verifies that
all its \T withdraw/ messages have been acknowledged, to preclude
interference by delayed messages.

\subsubsection{The inner protocol} \label{inner}

The inner protocol serves to ensure the resource allocation condition
\S Rq2/.  Every process forms in line 25 a set \S need/ of processes
from which it needs ``permission'' to enter \S CS/.  As the condition
$q\B\ in /\S CS/$ is now equivalent to $q\B\ at /27$, condition \S
Rq2/ is implied by the invariants
\begin{tab}
\S Jq0:/ \> $ r\in\S need/.q \Implies q \B\ at / 26 \Land r\in\S nbh0/.q $ ,\\
\S Rq2a:/ \> $ r\in\S nbh0/.q \Land q\in\S nbh0/.r $\\
\> $ \Implies r\in\S need/.q \Lor q\in\S need/.r 
\Lor \S job/.q * \S job/.r $ .
\end{tab}
We postpone the proofs of these invariants to Section \ref{mx_proof}.

At this point, we break the symmetry. Recall that we represent the
processes by natural numbers, and that, if $q<r$, we say that process
$q$ is \emph{lower} and that $r$ is \emph{higher}.  
Notifications from lower processes are regarded as requests for
permission that must be granted when possible, because we give priority
to lower processes. Therefore, when process
$p$ receives \T notify/ from $q < p$, it stores $q$ in $\S prom/.p$.
When the alternative \T prom/ is enabled, process $p$ grants
permission by sending \T gra/ to $q$.  In $\S away/.p$, it records the
lower processes to which it has granted permission. If it is at line
26 and in conflict with $q$, it puts $q$ in $\S need/.p$.

There is a difference in the interpretation of $\S need/.p$ for lower
and higher processes. If $q < p$, then $q\in\S need/.p$ means that
process $p$ is in conflict with $q$ and has granted priority to
$q$. Process $p$ therefore needs to wait for $q$'s \T withdraw/
message. If $p<q$, then $q\in\S need/.p$ means that process $p$ has
requested permisssion from $q$ and is still waiting for the \T gra/
message (no conflict implied).

\subsubsection{The middle layer} \label{middle}

Without waiting at line 25, the algorithm of Figure \ref{fifo} would
satisfy \S Rq2/, but it would have two defects. At line 26, one low
process could repeatedly pass all higher conflicting neighbours. Also,
long waiting queues of conflicting processes could form.  These
defects are treated by the middle layer.

When process $p$ enters at line 24, it assigns to $\S prio/.p$ the set
of processes with known incompatible \S job/s.  This set is finite
because it is contained in the finite set $\{ q \mid \S copy/.p(q)\ne
\S none/\}$.  Process $p$ then waits for the set $\S prio/.p$ to
become empty. It removes $q$ from $\S prio/.p$ when it receives \T
withdraw/ from $q$. In this way, the middle layer only admits
processes to the inner protocol that are not known to be conflicting
with processes in the inner protocol.  This improves the performance
by making it unlikely that at line 26 long waiting queues of
conflicting processes are formed. On the other hand, it ensures
starvation freedom. In fact, when process $p$ has executed line 25 and
its \T notify/ messages have arrived, any conflicting neighbour of $p$
that passes $p$ at line 26, will have to wait for $p$ at line 25, and
hence cannot pass $p$ again.

\begin{remark} 
  The first ideas for the present paper were tested in \cite{whh464}
  in a context with a single resource.  There, the \T notify/ messages
  are sent in the analogue of line 24 instead of line 25. This is also
  possible here. It has the effect that processes at line 25 are
  waiting for processes that arrived earlier at line 25. In other
  words, it induces a form of a first-come-first-served order. This is
  not a good idea for resource allocation. Consider, e.g., the
  following senario.

  Process $p_0$ arrives and starts using resource $r_0$ in \S CS/. Now
  processes $p_k$ for $k\geq 1$ arrive in their natural order at line
  24 at intervals $>\Delta$ (see Section \ref{async}), needing the
  resources $r_{k-1}$ and $r_k$, and with empty \S wack/. If the
  notifications are sent in line 24, they all remain waiting at line
  25, because $p_{k-1}\in\S prio/.p_k$, until $p_0$ has passed \S
  CS/. In the present version, with notifications sent at line 25, the
  processes with $k$ odd start waiting at line 25, while the processes
  with $k$ even go through to \S CS/.
\end{remark}

\subsection{The job model} \label {jobmodel}

For the central algorithm, jobs are abstract objects with a
compatibility relation. We need a job model for the registration
algorithm.

In the simple case of exclusive access, one may regard every job as a
set of resources, and define jobs to be compatible if and only if
these sets are disjoint.  In this case, we have $u*v \equiv (u\cap v =
\emptyset)$ and $\S none/ = \emptyset$.  If one wants to distinguish
read requests from write requests, however, one needs a more
complicated job model.  In this case, one could model jobs as pairs of
sets of resources, say $(r, w)$ where $r$ is the set of the resources
for read access and $w$ the set of resources for write access.  Jobs
$(r_1, w_1)$ and $(r_2, w_2)$ are then compatible if and only if the
intersections $w_1\cap w_2$, $w_1\cap r_2$, and $r_1\cap w_2$ are
empty.  One can also propose compatibility relations with more than
two permission levels, where ``shallow'' access (e.g.\ reading of
metadata) is allowed concurrently with ``innocent'' writing.

We therefore use a flexible job model that allows an arbitrary number
$K \geq 1$ of levels.  Let $\S upto/(K)$ be the set $\{i\in\Nat\mid i
\leq K\}$.  Let \S Rsc/ be the set of resources. We characterize a job
$u$ as a function $\S Rsc/\to\S upto/(K)$, and define compatibility of
jobs $u$ and $v$ by requiring that $u+v$ is at most $K$:
\begin{tabn} \label{defCompatible}
\> $ u * v \EQ (\all c\in \S Rsc/: u(c) + v(c) \leq K) $ .
\end{tabn}
In this way, relation $*$ is indeed symmetric, and the job \S none/
given by $\S none/(c) = 0$ for all $c$ is compatible with all jobs.

The simple job model is the case with $K=1$.  We take $K=2$ for the
readers/writers problem. Read access at resource $c$ requires $u(c)
\geq 1$, write access requires $u(c) = 2$.  In this way, concurrent
reading is allowed, while writing needs exclusive access.

\subsection{The registration algorithm} \label{queryalg}

The registration algorithm serves to provide the processes with
upperbound sets \S nbh/.  The idea is that every process registers for
the resources that it needs, and that it then obtains lists of other
registered clients.  For the sake of flexibility, we assume that the
resources are distributed over \emph{sites} by means of a function $\S
loc/:\S Rsc/ \to \S Site/ $ from resources to sites.  The sets \S Rsc/
and \S Site/ are supposed to be finite.  In order to preclude that
they become bottlenecks, the sites get only the task to maintain a
registration list and to answer to queries.

A process with a high registration level needs to perform much
communication.  We therefore offer the processes the option to lower
their registration level, concurrently with the other activities.

We use the job model with upper bound $K$ of Section \ref{jobmodel}.
In particular, every job is a function $\S Rsc/\to\S upto/(K)$.  A
process can only use resource $c$ at level $k$ if it is registered at
site $\S loc/(c)$ for level $\geq k$.  It therefore has an array \S
fun/ such that $\S fun/.p(s)$ is the $p$'s registration level at site
$s$.  When some process obtains a new \S job/, it needs at site $s$
the level
\begin{tab}
\> $ L(\S job/)(s) = \R Max/ \{\S job/(c) \mid \S loc/(c) = s\}$ .
\end{tab}
For functions $f$, $g:\S Site/\to\Nat$, we define $f\leq g$ to mean
$(\all s: f(s) \leq g(s))$.  

In line 21 of Figure \ref{fifo}, the environment gives process $p$ a
new job.  At line 22, process $p$ may have to wait, to avoid
interference with the lowering thread that is treated below.  It then
sends $L(\S job/)(s)$ to site $s$ if it is positive, and thus asks the
site for a lists of clients that might compete for its resources.  The
set \S nbh/ gets its contents while the process waits at line 23,
through messages from the sites in answer to \T asklist/.

We give the sites very small tasks.  Site $s$ communicates with the
processes by receiving messages \T asklist/ and \T lower/, and
answering by \T answer/ and \T done/, respectively; this
according to the code:
\begin{tab}
\>\+\+ $\B site/(s)\B\ from/(q):$\\
$\bar$ \> $ \T asklist/.q.s\;?\:k $ ;
 $ \S list/(q) := \R max/(\S list/(q), k) $ ;\\
\> $ \T answer/.s.q\:!\:\{r \mid \S list/(r) > K-k \} $ .\\
$\bar$ \> $ \T lower/.q.s\;?\:k $ ; $ \S list/(q):= k $ ;
 $ \T done/.s.q\:! $ .\-\\
\B end site/ .
\end{tab}
In this code, \S list/ is the private extendable array $\S list/.s$ of
$s$, that holds the levels of registered clients.  The value 0 means
not-registered.  The answering messages \T answer/ contain the
processes that are in potential conflict at the level $k$.  If process
$q$ lowers at site $s$ its level to $k$, it gets response \T done/ as
an acknowledgment.

Process $p$ receives the messages from site $s$ according to:
\begin{tab}
\>\+\+ $ \B listen/(p)\B\ from/(s): $ \\
$ \bar $ \> $ \T answer/.s.p\:?\: v$ ; 
 $ \S nbh/:=\S nbh/\cup(v\setminus \{p\}) $ ;\\
\> $ \B if /\S fun/(s) < L(\S job/)(s) \B\ then/ $ \\
\>\>  $ \S pack/:=\S pack/\cup(v\setminus \{p\}) $ ;\\
\>\> $ \S fun/(s) := L(\S job/)(s) \B\ endif/ $ ;\\
\> $ \S curlist/ := \S curlist/\setminus \{s\} $ .\\
$ \bar $ \> $ \T done/.s.p\:? $ ; 
$ \S reglist/ := \S reglist/\setminus \{s\} $ .\-\\
\B end listen/ .
\end{tab}
If process $p$ increases its registration level at site $s$, it
collects the potential competitors in the private variables \S pack/.
When it has received all answers, its sends all members of \S pack/ a
message \T hello/ in line 23, and waits for the responses \T welcome/
at line 24.  The reason for this is that the members of \S pack/ can
be anywhere in their protocol and need not have $p \in \S nbh/$.

The new messages \T hello/ and \T welcome/ are between processes.
They are treated in the following two alternatives that should be
included in \B receive/ of Figure \ref{fifo}.
\begin{tab}
\>\+ $ \bar $ \> $ \T hello/.q.p\:? $ ; \\
\> $ \T welcome/.p.q\:!\:(\S pc/ \geq 26 \Land q\notin \S nbh/\:?\:
\S job/ : \S none/) $ ;\\
\> $ \B if / \S pc/\geq 23\B\ then /
\B add /q\B\ to /\S nbh/\B\ endif/ $ .\\
$ \bar $ \> $ \T welcome/.q.p\:?\: v $ ; $ \B remove /q\B\ from /\S pack/$ ;\\
\> $ \B if /v\ne\S none/\B\ then / \S copy/(q) := v\B\ endif/ $ .
\end{tab}
If process $p$ is in $\{26\dots\}$ and $q\notin\S nbh/.p$, the message
\T welcome/ carries the \S job/ of $q$ as a belated notification.
Otherwise, it only holds \S none/ as an acknowledgement.  Here, we use
a conditional expression of the form $(b\,?\,x:y)$ to mean $x$ if $b$
holds and otherwise $y$, as in the programming language C.

In \T welcome/, the assignment to $\S copy/.p(q)$ ensures that, when
process $p$ raises its registration level, it cannot enter its inner
protocol when in conflict with $q$, while $q$ remains in its inner
protocol.  At this point, the guard of line 25 is necessary for
safety.  This is a third reason for the middle layer of Section
\ref{middle}.

If process $p$ receives \T hello/ and is in $\{23\dots 25\}$, it adds
$q$ to \S nbh/ to notify it at line 25.  If process $p$ sends its \S
job/ to $q$, it adds $q$ to \S nbh/ because the \S job/ must be
withdrawn later.  At this point, the set \S nbh0/ can become a proper
subset of \S nbh/.

Lowering means the choice of a new value \S news/ for \S fun/, which
can be equal or lower than the current value.  The processes can lower
at the sites more or less concurrently with the loop 21--28. For this
purpose, each of them gets a separate concurrent thread with a
separate process counter \S pcr/. We write $q\B\ at /\ell$ to mean $\S
pc/.q=\ell$ if $\ell\in\{21\dots 28\}$, and to mean $\S pcr/.q=\ell$
if $\ell\in\{31\dots 33\}$.  If $\ell\in\{21\dots 28\}$, we write
$q\B\ in /\{\ell\dots\}$ to mean $\S pc/.q\geq\ell$.  Initially, every
process is both \B at/ 21 and \B at/ 31.

\begin{tab}
  \>\+ $ \B lowering/(p): $\\
  \B loop/\\
  31: \>\> $ \B choose /\S news/\B\ with / \S news/\leq \S fun/ $ ;\\
  32: \>\> $ \B await / \S pc/ = 21\Lor
  (\S pc/ \geq 25\Land L(\S job/) \leq \S news/)$ ;\\
  \>\> $ \S reglist/ := \{s \mid \S news/(s) \ne \S fun/(s)\} $ ;
  $ \S fun/:= \S news/ $ ; \\
  \>\> $ \B for all /s \in\S reglist/\B\ do /
  \T lower/.p.s\:!\:\S news/(s) \B\ od/ $ ; \\
  33: \>\> $ \B await /\S reglist/ = \emptyset $ ;\\
  \B endloop/ .
\end{tab}

The lowering thread shares the variable \S fun/ with the main thread,
and has private variables \S news/ and \S reglist/.  It is idle at
line 31.  When it is idle, to replace \S fun/, the environment can
give the process a value \S news/.  At line 32, the process informs
the sites $s$ for which the level is to be modified by sending them a
message \T lower/ with the new value.  If the main thread is in
22--24, the lowering thread needs to wait to avoid interference.  The
guard $L(\S job/) \leq \S news/$ is needed to protect the current job
of $p$. 

Initially, the private sets $\S pack/.p$ and $\S reglist/.p$ are empty
and the functions $\S fun/.p$, $\S news/.p$, and $\S list/.s$ are
constant zero.

\subsection{Aborting the entry protocol} \label{s.abort}

As announced we give the environment the option to abort the entry
protocol under certain circumstances:

\begin{tab}
  \> $ \B abort/(p): $ \+\+\\
  $ \bar\quad \S pc/ = 24 \Land \S pack/ = \emptyset
  \TO \S job/ := \S none/\;;\;
  \S nbh/ := \emptyset\;;\; \S pc/ := 21 $ .\\
  $ \bar\quad \S pc/ = 25 \TO \S job/ := \S none/ $ ; 
  $\S nbh/ := \emptyset$ ; $\S prio/ := \emptyset$ ; $\S pc/ := 21 $ .\\
  $ \bar\quad \S pc/ = 26 \Land
  \S need/\cap\{q\mid p < q\} = \emptyset \TO$ \\
  \>\> $ \B for all /q\in\S nbh/ \B\ do /\T withdraw/.p.q\:!\B\ od/ $ ;
  $\S wack/ := \S nbh/ $ ; \\
  \>\> $ \S need/ := \emptyset $ ; $ \S job/ := \S none/\;;\;\S nbh/
  := \S nbh0/ := \emptyset$ ; $ \S pc/ := 21 $ .
\end{tab}

When the environment of $p$ wants to abort the entry protocol at line
24, it may have to wait for emptiness of \S pack/.  This waiting is
necessary to catch the expected \T welcome/ messages.  It is shorter
than $2\Delta$ (see Section \ref{async}).  At line 26, it may have to
wait for emptiness of the higher part of \S need/, necessary to catch
the expected \T gra/ messages.  This waiting is short because process
$p$ has priority over its higher neighbours. Indeed, the higher part
of \S need/ is empty after $\Gamma + 2\Delta$ (see Section
\ref{async}).

To summarize, every process $p$ has four concurrent threads that
execute their atomic steps in an interleaved way.  These threads are:
the \emph{environment} with the prompting steps at lines 21 and 31 and
the 3 aborting steps, the \emph{forward} thread with the steps at the
lines 22--28; the \emph{lowering} thread with the steps at the lines
32--33; and the \emph{triggered} thread for the 6 messages from other
processes, the 2 messages from sites, and the delayed answers \T
after/ and \T prom/.  Moreover, every site has two commands for the
message keys \T asklist/ and \T lower/.

\section{Verification of Safety}
\label{algorithm}

In this section, we prove the safety properties of the algorithm.  For
this purpose, we model the algorithm as a transition system that is
amenable to formal verification.

The modelling, in particular of the asynchronous messages, is
discussed in Section \ref{channels}. In Section \ref{safety}, we
describe the verification of safety by means of invariants.  In
Section \ref{peculiarity}, we describe some choices we made to ease
our proof management.

We first treat the central algorthm of Figure \ref{fifo} and prove
that it satisfies its requirement \S Rq2/.  In Section \ref{InvOuter}
we develop the invariants of the outer protocol that are needed in the
proof of \S Rq2/.  As indicated in Section \ref{inner}, the mutual
exclusion predicate \S Rq2/ is implied by \S Jq0/ and \S Rq2a/. These
invariants of the inner protocol are proved in Section \ref{mx_proof},
together with a number of auxiliary ones.

In Section \ref{queryadd}, we add the registration algorithm, verify
that the new messages are modelled correctly and that this addition
does not disturb the safety properties of the central algorithm.
Section \ref{queryverif} shows that it indeed serves its purpose and
guarantees the invariant \S Rq1/.

In Section \ref{invarProgress}, we conclude the list of invariants by
presenting some invariants that are needed in the proofs of progress
in the Sections \ref{nodead} and \ref{proofThm1}.

\subsection{Modelling of messages} \label {channels}

We model the algorithm as a transition system with as state space the
Cartesian product of the private state spaces augmented with the
collection of messages in transit.  The messages in transit are
modelled by shared variables.

There are two kinds of messages: messages of type \T void/ without
content but consisting of a single key word, and messages with content.

For a message key $m$ of type \T void/, we use $m.q.r$ as an integer
shared variable that holds the number of messages $m$ in transit from
$q$ to $r$, to be inspected and modified only by the processes $q$ and
$r$.

A sending command $m.p.q\,!$ from $p$ to $q$ e.g., as used in Figure
\ref{fifo}, corresponds to an incrementation of $m.p.q$ by one, which
can be denoted $ \ m.p.q\T ++/ \ $.  A receiving command $m.q.p\,?$
from $q$ by $p$ followed by a statementlist $S$ corresponds to a
guarded alternative in which $m.q.p$ is decremented by one when
positive:
\begin{tab}
\> $\bar\quad m.q.p > 0 \TO m.q.p\T --/ \;;\; S$ .
\end{tab}
The value of $m.q.r$ can be any natural number, but in our algorithm,
we preserve the invariants $m.q.r \leq 1$: there is always at most one
message $m$ in transit from $q$ to $r$.  Initially, no messages are in
transit: $m.q.r=0$ for every message key $m$.

For messages with content, like \T notify/, the above way of modelling
cannot be used.  In principle, we should model such messages by means
of bags (multisets) of messages from sender to destination. In the
algorithm, however, there is never more than one message with key $m$
in transit from $q$ to $r$. For simplicity, therefore, we model such a
channel with key $m$ from $q$ to $r$ as a shared variable $m.q.r$,
which equals $\bot$ if and only if there is no message $m$ in transit
from $q$ to $r$.

In particular, we model the sending command $\T notify/.p.q\,!\:\S
job/$ from $p$ to $q$ at line 25 therefore as $\T notify/.p.q := \S
job/$.  Reception of \T notify/ from $q$ by $p$ followed by $S$ is
modelled by
\begin{tab}
\> $\bar\quad \T notify/.q.p \ne \bot \TO $\+\+ \\
$ \S copy/(q) := \T notify/.q.p \;;\; \T notify/.q.p:= \bot\;;\; S $ .
\end{tab}
Initially, $\T notify/.q.r = \bot$ for all $q$ and $r$. 

As we model the bag by a single variable, we need to make sure that
the bag has never more than one element.  In other words, this way of
modelling gives us the proof obligation that $\T notify/.q.r$ is sent
only under the precondition $\T notify/.q.r = \bot\,$.  This will
follow from the invariant \S Iq7a/ of Section \ref {InvOuter} below.

\subsection{Using invariants} \label{safety}

In a distributed algorithm, at any moment, many processes are able to
do a step that modifies the global state of the system. The only way
to reason successfully and reliably about such a system is to analyse
the properties that cannot be falsified by any step of the
system. These are the invariants.

Formally, a predicate is called an \emph{invariant} of an algorithm if
it holds in all reachable states.  A predicate $J$ is called
\emph{inductive} if it holds initially and every step of the algorithm
from a state that satisfies $J$ results in a state that also satisfies
$J$.  Every inductive predicate is an invariant. Every predicate
implied by an invariant is an invariant.

When a predicate is inductive, this is often easily verified. In many
cases, the proof assistant PVS is able to do it without user
intervention. It always requires a big case distinction, because the
transition system has many different alternatives.

Most invariants, however, are not inductive.  Preservation of such a
predicate by some alternatives needs the validity of other invariants
in the precondition. We use PVS to pin down the problematic
alternatives, but human intelligence is needed to determine the useful
other invariants.

In proofs of invariants, we therefore use the phrase
``\emph{preservation of $J$ at $\ell_1\dots \ell_m$ follows from
  $J_1\dots J_n$}'' to express that every step of the algorithm with
precondition $J\land J_1\dots J_n$ has the postcondition $J$, and that
the additional predicates $J_1\dots J_n$ are only needed for the
alternatives $\ell_1\dots \ell_m$.

We use the following names for the alternatives.  The first 8
alternatives of \B central/ in Figure \ref{fifo} are indicated by the
line numbers. The alternatives of \B receive/ are indicated by the
message names and the labels \T after/ and \T prom/. We indicate the
aborting alternatives of Section \ref{s.abort} by \T ab24/, \T ab25/,
\T ab26/.

For all invariants postulated, the easy proof that they hold initially
is left to the reader. We use the term invariant in a premature
way.  See the end of this section.  

\subsection{Proof engineering} \label{peculiarity}

Effective management of the combined design and verification of such
an algorithm requires a number of measures.  We give most invariants
names of the form \S Xqd/, where \S X/ stands for an upper case letter
and \S d/ for a digit. This enables us to reorder and rename the
invariants in the text and the PVS proof files and to keep them
consistent.  Indeed, any modification of proof files must be done very
carefully to avoid that the proof is destroyed.  Using short
distinctive names also makes it easy to search for definitions and to
see when all of them have been treated.

Line numbers may change during design.  In order to use query-replace
for this in all documents, we use line numbers of two digits.  In this
way, we preclude that the invariants get renamed by accident.  This is
also the reason to use disjoint ranges for the line numbers of Figure
\ref{fifo} and the lowering thread of Section \ref{queryalg}.

There is a trade off in the size of the invariants.  Smaller
invariants are easier to prove and easier to apply, but one needs more
of them, and they are more difficult to remember.  We therefore often
combine a number of simple properties in a single invariant, see \S
Iq1/ below.  Bigger invariants are sometimes needed to express
different aspects of a complicated state of affairs, compare \S Iq2/
below.

\subsection {Invariants of the outer protocol}
\label{InvOuter} 

For now, we restrict ourselves to the transition system with 14
transitions of Figure \ref{fifo}.  The nine steps of Section
\ref{queryalg} are added in Sections \ref{queryadd} and
\ref{queryverif}.  The three steps of Section \ref{s.abort} are added
in Section \ref{invarProgress}.

We have two invariants about neighbourhoods:
\begin{tab}
\S Iq0:/ \> $ q\notin \S nbh/.q$ ,\\
\S Iq1:/ \> $ r\in \S nbh0/.q \Implies q\B\ in /\{26\dots\} 
\Land r\in \S nbh/.q$ . 
\end{tab}
These predicates are easily seen to be inductive.

At line 24, the processes wait for acknowledgements as expressed by
emptiness of \S wack/. This corresponds to the invariant:
\begin{tab}
\S Iq2:/ \> $ \T withdraw/.q.r + |\,q\in\S after/.r\,| + \T ack/.r.q
= |\,r\in\S wack/.q\,| $ .
\end{tab}
Recall that $\T withdraw/.q.r$ is the number of \T withdraw/ messages
from $q$ to $r$ and that $\T ack/.r.q$ is the number of \T ack/
messages from $r$ to $q$.  For Boolean $b$, we define $|\,b\,|\in\Nat$
to be 1 if $b$ holds, and 0 otherwise.  Predicate \S Iq2/ is a concise
expression of a complicated fact.  Namely, $r\in\S wack/.q$ holds if
and only if if there is a \T withdraw/ message in transit from $q$ to
$r$, or an \T ack/ message in transit from $r$ to $q$, or $q\in\S
after/.r$.  Furthermore, the three possibilities are mutually
exclusive. Finally, there is at most one \T withdraw/ message from $q$
to $r$, and at most one \T ack/ message from $r$ to $q$.  One could
therefore split \S Iq2/ into 9 different invariants.

Preservation of \S Iq2/ when \T withdraw/ is sent at line 28 follows
from the inductive invariant:
\begin{tab}
\S Iq3:/ \> $ q\B\ in /\{25\dots\} \Implies \S wack/.q = \emptyset $ .
\end{tab}

For practical purposes, it is useful to notice that \S Iq2/ and \S
Iq3/ together imply
\begin{tab}
\S Iq2a:/ \> $ q\B\ in /\{25\dots\} \Implies \T withdraw/.q.r = 0 \Land 
q \notin\S after/.r $ .
\end{tab}

As announced, one of the functions of the outer protocol is to
guarantee that, under suitable conditions, process $r$ has the job of
$q$ in its variable $\S copy/.r(q)$.  In fact, the conditions are that
$r$ is in $\S nbh0/.q$ and that there is no message \T notify/ in
transit from $q$ to $r$, as expressed in the invariant
\begin{tab}
\S Iq4:/ \> $ r \in\S nbh0/.q \Land \T notify/.q.r = \bot 
 \Implies \S copy/.r(q) = \S job/.q $ .
\end{tab}

Preservation of \S Iq4/ at line 21 follows from \S Iq1/.  Preservation
at \T after/ follows from \S Iq1/ and \S Iq2a/.  Preservation at line
25 and \T notify/ follows from \S Iq1/ and the new invariants:
\begin{tab}
\S Iq5:/ \> $ \S job/.q = \S none/ \EQ q\B\ at / 21 $ ,\\
\S Iq6:/ \> $ q \B\ in /\{26\dots\} \Land
\T notify/.q.r \ne \bot \Implies \T notify/.q.r = \S job/.q $ .
\end{tab}
Predicate \S Iq5/ is inductive.  Preservation of \S Iq6/ at line 25
follows from the new invariant
\begin{tab}
\S Iq7a:/ \> $ q\B\ at /25 \Implies
\T notify/.q.r = \bot\Land \S copy/.r(q) =\S none/ $ .
\end{tab}
Predicate \S Iq7a/ is logically implied by \S Iq2/, \S Iq3/, and the new
invariant:
\begin{tab}
\S Iq7:/ \> $ (\T notify/.q.r = \bot\Land \S copy/.r(q) =\S none/)
\Lor (q \B\ in /\{26\dots\} \Land r\in\S nbh/.q) $\\
\> $  \Lor\, \T withdraw/.q.r > 0 \Lor q\in\S after/.r $ .
\end{tab}
Preservation of \S Iq7/ at \T after/ follows from the new invariant:
\begin{tab}
\S Iq8:/ \> $ \T notify/.q.r = \bot \Lor \S copy/.r(q) = \S none/$ .
\end{tab}
Preservation of \S Iq8/ at line 25 follows from \S Iq7a/. 

This is not circular reasoning: the above argument shows that, if all
predicates \S Iq*/ hold in the precondition of any step, they also
hold in the postcondition. Therefore, the conjunction of them is
inductive, and each of them is an invariant.

\subsection{The proof of mutual exclusion} \label{mx_proof}

In Section \ref{inner}, we saw that the mutual exclusion predicate \S
Rq2/ is implied by \S Jq0/ and \S Rq2a/.  In this section, we prove
that these two predicates are invariants.

Preservation of predicate \S Jq0/ at \T prom/ follows from \S Iq4/, \S
Jq0/, and the new invariants
\begin{tab}
  \S Rq1a:/ \> $ q\B\ in /\{26\dots\} \Land r\B\ in /\{26\dots\}
  \Implies q = r \Lor r\in\S nbh0/.q \Lor \S job/.q * \S job/.r $ ,\\
  \S Jq1:/ \> $ q\in\S prom/.r \Implies q < r $ ,\\
  \S Jq2:/ \> $ q < r\Implies |\,\T notify/.q.r \ne\bot \,| +
  |\,q\in\S prom/.r\,| + \T gra/.r.q = |\, r\in\S need/.q\,| $ .
\end{tab}
Predicate \S Rq1a/ is a strengthening of \S Rq1/ of Section
\ref{safetyreq} that is guaranteed by the registration algorithm.
This is shown in Section \ref{queryverif}.  Predicate \S Jq1/ is
inductive.  Preservation of \S Jq2/ at 25 and \T prom/ follows from \S
Iq5/, \S Jq0/ and \S Jq1/.  Note the similarity of \S Jq2/ with \S
Iq2/.

Predicate \S Rq2a/ of Section \ref{inner} is implied by \S Iq0/, \S
Iq1/, \S Iq2a/, and the new invariants:
\begin{tab}
\S Jq3:/ \> $ q < r \Land  r\in\S nbh0/.q
\Implies r\in\S need/.q\Lor q \in\S away/.r $ ,\\
\S Jq4:/ \> $ q\in\S away/.r \Land q\in \S nbh0/.r 
\Land \T withdraw/.q.r = 0 $\\
\> $ \Implies q\in\S need/.r\Lor \S job/.q *\S job/.r$ .
\end{tab}

Preservation of \S Jq3/ at \T withdraw/ follows from \S
Iq1/, \S Iq2a/.  At \T gra/, it follows from the new invariant 
\begin{tab}
  \S Jq5:/ \> $ \T gra/.r.q > 0 \Implies q\in\S away/.r $ .
\end{tab}

Preservation of \S Jq4/ at 21 follows from \S Iq1/.  Preservation at
\T prom/ follows from \S Iq1/, \S Iq4/, \S Jq0/, \S Jq1/, and \S Jq2/.
Preservation at line 25 and at \T gra/ and \T withdraw/ follows from
\S Iq4/, \S Jq5/ and the new invariants:
\begin{tab}
  \S Jq6:/ \> $ q \in\S away/.r \Implies
  q < r \Land \T notify/.q.r = \bot $ ,\\
  \S Jq7:/ \> $ q \in\S away/.r \Land \T withdraw/.q.r = 0 \Implies
  r\in\S nbh0/.q $ .
\end{tab}

Preservation of \S Jq5/ at \T withdraw/ follows from \S Iq2a/, \S
Jq0/, and \S Jq2/.  Preservation of \S Jq6/ follows at line 25 from \S
Iq1/, \S Iq2/, \S Iq3/, and \S Jq7/, and at \T prom/ from \S Jq1/, \S
Jq2/.  Preservation of \S Jq7/ at 25 and 28 follows from \S Iq1/, and
at \T prom/ and \T withdraw/ from \S Jq0/, \S Jq1/, \S Jq2/, and \S
Jq6/.

This concludes the proof of the invariants \S Jq0/ and \S Rq2a/ under
assumption of \S Rq1a/.

\subsection {Adding registration} \label{queryadd}

We now add the registration algorithm to the central algorithm, i.e.,
we extend the transition system with the nine transitions of Section
\ref{queryalg}.  There are three things to verify. The modeling must
be correct, the proof of the central algorithm must not be disturbed,
and condition \S Rq1a/ of Section \ref{mx_proof} must be guaranteed.
The first two points are treated in this section.  The third point is
postponed to the next section.

The new messages \T asklist/, \T answer/, \T welcome/, and \T lower/
are not void, and are therefore modelled in the same way as \T
notify/.  This gives us the obligation to prove, for each of these
four message keys $m$, that the value of $m$ is $\bot$ whenever a
message $m$ is sent.  This follows from the invariants (similar to \S
Iq2/):
\begin{tab}
  \S Kq0:/ \> $|\, \T asklist/.q.s\ne\bot\,| + |\,\T answer/.s.q
  \ne\bot\,| = | \, s\in\S curlist/.q\, | $ ,\\
  \S Kq1:/ \> $ \T hello/.q.r + |\,\T welcome/.r.q\ne\bot\,| = 
  |\, q\B\ at /24\Land r\in\S pack/.q\, | $ ,\\
  \S Kq2:/ \> $|\, \T lower/.q.s\ne\bot\,| + \T done/.s.q
  = | \, q \B\ at /33 \Land s\in\S reglist/.q\, | $ .
\end{tab}
Predicate \S Kq0/ is preserved at 22 because of the inductive invariant:
\begin{tab}
\S Kq3:/ \> $ q\B\ at /23\Lor \S curlist/.q  = \emptyset $ .
\end{tab}
The invariants \S Kq0/ and \S Kq3/ together imply
\begin{tab}
\S Kq0a:/ \> $ \T answer/.s.q\ne\bot \Implies q\B\ at /23 $ .
\end{tab}
Predicate \S Kq1/ is preserved at \T answer/ because of \S Kq0a/.
Predicate \S Kq2/ is inductive.

We turn to the question whether the central algorithm is disturbed by
the new registration commands.  The only variables that the
registration algorithm shares with the central algorithm are \S pc/,
\S pcr/, \S nbh/, \S pack/, and \S copy/.  Sharing \S pc/ is harmless,
because the flow of controle is not modified.  Sharing \S pack/ and \S
pcr/ is harmless because the central algorithm has no invariants for
\S pack/ and \S pcr/.  Sharing \S nbh/ is almost harmless because all
invariants except \S Iq0/ allow enlarging \S nbh/.  Predicate \S Iq0/
is preserved by \T hello/ because of $\T hello/.q.q=0$ which follows
from \S Kq1/ together with the inductive invariant
\begin{tab}
  \S Kq4:/ \> $ q\notin \S pack/.q $ .
\end{tab}

Modification of \S copy/ by \T welcome/ requires new invariants.
Predicate \S Iq4/ is preserved at \T welcome/ because of \S Iq1/ and
the new invariant
\begin{tab}
\S Kq5:/ \> $ q\B\ in /\{26\dots\} \Implies 
\T welcome/.q.r \in \{ \bot, \S none/, \S job/.q \} $ .
\end{tab}
Predicate \S Iq7/ is preserved at \T welcome/ because of 
\begin{tab}
\S Kq6:/ \> $ \T welcome/.q.r \in \{ \bot, \S none/ \}
\Lor \T withdraw/.q.r > 0 $ \\
\> $ \Lor\: q\in\S after/.r
\Lor (q\B\ in /\{26\dots\} \Land r\in\S nbh/.q) $ .
\end{tab}
Predicate \S Iq8/ is preserved at \T welcome/ because of 
\begin{tab}
\S Kq7:/ \> $ \T welcome/.q.r \in \{ \bot, \S none/\}
\Lor (\T notify/.q.r = \bot \Land \S copy/.r(q) = \S none/) $ .
\end{tab}

Predicate \S Kq5/ is preserved at 25 because of \S Iq2a/ and \S Kq6/.
Predicate \S Kq6/ is preserved at \T after/ and \T hello/ because of
\S Iq1/ and \S Kq7/.  Predicate \S Kq7/ is preserved at 25 and \T
hello/ because of \S Iq2a/, \S Iq7/, and \S Kq6/.

This concludes the proof that the central algorithm in combination
with the registration algorithm preserves the invariant \S Rq2/ 
of Section \ref{safetyreq}. 

\subsection {Safety of registration} \label{queryverif}

We approach predicate \S Rq1a/ in a bottom-up fashion.  The lowering
thread does not interfere with the main thread because of the
inductive invariants:
\begin{tab}
  \S Lq0:/ \> $ q\B\ in /\{23, 24\} \Implies q\B\ in /\{31, 32\} $ ,\\
  \S Lq1:/ \> $ \S news/.q \leq \S fun/.q  $ .
\end{tab}
The sets \S prio/ and \S pack/ are only nonempty in specific
locations, as expressed by the invariants:
\begin{tab}
  \S Lq2:/ \> $ q\in\S prio/.r \Implies r\B\ at /25 
  \Land q\notin\S after/.r $ ,\\
  \S Lq3:/ \> $ q\B\ in /\{23, 24\} \Lor \S pack/.q = \emptyset $ .
\end{tab}
Indeed, \S Lq2/ is inductive.  Predicate \S Lq3/ is preserved by \T
answer/ because of \S Kq0a/.

For the communication with the sites, we claim the invariants:
\begin{tab}
  \S Lq4:/ \> $ \T asklist/.q.s \in \{\bot, L(\S job/.q)(s) \} $ ,\\
  \S Lq5:/ \> $ \T lower/.q.s \in \{ \bot, \S fun/.q(s)\} $ .
\end{tab}
Predicate \S Lq4/ is preserved at 21 and 28 because $\T asklist/.q.s
\ne\bot$ implies $q\B\ at /23$, as follows from \S Kq0/ and \S Kq3/.
Predicate \S Lq5/ is preserved at \T answer/ because of \S Kq0a/, \S
Kq2/, and the mutual exclusion invariant \S Lq0/.

There are subtle relations between $L(\S job/.q)(s)$, $\S fun/.q(s) $,
and $\S list/.s(q)$ expressed by the invariants:
\begin{tab}
  \S Lq6:/ \> $ q\B\ at / 22\Lor s\in\S curlist/.q
  \Lor L(\S job/.q)(s) \leq \S fun/.q(s) $ ,\\
  \S Lq7:/ \> $ q\B\ in /\{23\dots\} \Land \T asklist/.s.q = \bot
  \Implies L(\S job/.q)(s) \leq \S list/.s(q) $ ,\\
  \S Lq8:/ \> $\S fun/.q(s) \leq \S list/.s(q) $ .
\end{tab}
Predicate \S Lq6/ is preserved at line 32 because of \S Iq5/.
Predicate \S Lq7/ is preserved by \T asklist/ because of \S Lq4/.  It
is preserved by \T lower/ because of \S Kq2/, \S Kq3/, \S Lq0/, \S
Lq5/, and \S Lq6/.  Predicate \S Lq8/ is preserved at 32 because of \S
Lq1/, at \T answer/ because of \S Kq0/, \S Kq3/, \S Lq7/, and at \T
lower/ because of \S Lq5/.


After this preparation, we turn to the proof that the registration
algorithm satisfies it purpose, i.e., guarantees predicate \S Rq1a/ of
Section \ref{mx_proof}.  At line 23, process $q$ expects and receives
an answer from site $s$.  This answer is a set that contains process
$r$ iff $L(\S job/.q)(s) + \S list/.s(r) > K$, where $\S list/.s$ is
the value at the time of sending the answer.  In this way, we arrive
at the invariant:
\begin{tab}
  \S Mq0:/ \> $ q\B\ in /\{23\dots\} \Implies q = r \Lor
  r \in\S nbh/.q \Lor \T hello/.r.q > 0 $\\
  \> $ \Lor (r\B\ at /23 \land q\in\S pack/.r) \Lor
  L(\S job/.q)(s) + \S fun/.r(s) \leq K$ \\
  \> $ \Lor (s\in\S curlist/.q\Land(\T answer/.s.q = \bot \Lor r\in\T
  answer/.s.q)) $ .
\end{tab}
Predicate \S Mq0/ is preserved at line 22, 32, \T welcome/, \T
asklist/ because of \S Kq0a/, \S Lq1/, \S Kq1/, \S Lq4/, \S Lq8/.  It
is preserved at \T answer/ because of \S Kq0a/ and the new invariant:
\begin{tab}
  \S Mq1:/ \> $ q\B\ in /\{23\dots \} \Land \T answer/.s.r\ne\bot $ \\
  \> $ \Implies q = r \Lor r \in\S nbh/.q \Lor q\in\T answer/.s.r $\\
  \> $ \Lor L(\S job/.q)(s) + L(\S job/.r)(s) \leq K$ \\
  \> $ \Lor (s\in\S curlist/.q\Land(\T answer/.s.q = \bot \Lor r\in\T
  answer/.s.q)) $ .
\end{tab}
Predicate \S Mq1/ is preserved at 21, 22, 28 because of \S Kq0a/.  It
is preserved by \T asklist/ because of \S Kq0/, \S Kq3/, \S Lq4/, \S
Lq7/.

The predicates \S Mq0/, \S Kq1/, \S Kq3/, \S Lq6/ together imply the
derived invariant:
\begin{tab}
  \S Mq0a:/ \> $ q\B\ in /\{24\dots\} \Land r\B\ in /\{24\dots\}
  \Land \S pack/.r = \emptyset $ \\
  \> $ \Implies q = r\Lor r \in\S nbh/.q \Lor \S job/.q * \S job/.r$ .
\end{tab}
Predicate \S Mq0a/ approximates \S Rq1a/, but the disjunct $r\in\S
nbh/.q$ of \S Mq0a/ is weaker than the disjunct $r\in\S nbh0/.q$ of \S
Rq1a/.  As a remedy, we postulate the invariant:
\begin{tab}
  \S Mq2:/ \> $ q\B\ in /\{26\dots\} \Land r\in \S nbh/.q $\\
  \> $ \Implies r\in \S nbh0/.q \Lor \T welcome/.q.r = \S job/.q
  \Lor \S copy/.r(q) = \S job/.q $ .
\end{tab}
Predicate \S Mq2/ is preserved at \T notify/, \T answer/, \T after/,
\T welcome/, \T hello/ because of \S Iq1/, \S Iq2a/, \S Iq5/, \S Iq8/,
\S Kq0a/, \S Kq1/, \S Kq5/.

The conjunction of \S Mq0a/ and \S Mq2/ is not strong enough to imply
\S Rq1a/.  We need yet another invariant. Indeed, predicate \S Rq1a/
is implied by \S Lq2/ and the new invariant
\begin{tab}
  \S Mq3:/ \> $ q\B\ in /\{26\dots\} \Land r\B\ in /\{25\dots\} $ \\
  \> $ \Implies q = r\Lor r \in\S nbh0/.q \Lor q\in\S prio/.r \Lor \S
  job/.q * \S job/.r$ .
\end{tab}
Predicate \S Lq2/ serves to eliminate the alternative $q\in\S prio/.r$
of \S Mq3/.  Here we see, as announced at the end of Section
\ref{queryalg}, that correctness of the registration algorithm relies
on the guard of line 25, the emptiness of \S prio/.

Predicate \S Mq3/ is preserved at lines 24 and 25 because of \S Mq0a/
and \S Iq2a/, \S Iq5/, \S Kq1/, \S Lq3/, \S Mq2/.  It is preserved at
\T welcome/ because of \S Iq2a/.

This concludes the proof of \S Rq1a/ of Section \ref{mx_proof} and
thus of \S Rq1/ and \S Rq0/ of Section \ref{safetyreq}.

\subsection{Invariants concluded} \label{invarProgress}

We conclude the section by developing some invariants needed for the
proofs of progress in the Sections \ref{nodead} and \ref{proofThm1}.
Primarily, we need additional invariants about \S need/, \S copy/, \S
prio/, because of the guards of the alternatives 25, 26, \T after/,
and \T prom/.

We need two invariants for \S need/ in the inner protocol:
\begin{tab}
\S Nq0:/ \> $ q < r \Land q\in\S need/.r \Implies q\in\S away/.r $ ,\\
\S Nq1:/ \> $ q < r \Land q\in \S need/.r \Land \S job/.q * \S job/.r
\Implies \T withdraw/.q.r > 0 $ .
\end{tab}
Predicate \S Nq0/ is inductive.  Predicate \S Nq1/ is preserved at 21
because of \S Iq5/, at 25 because of \S Iq1/, \S Iq4/, \S Jq6/ and \S
Jq7/, and at 28 because of \S Iq1/, \S Jq0/, \S Jq7/, and \S Nq0/.  It
is preserved at \T prom/ because of \S Iq4/, \S Jq0/, and \S Jq2/.

We also need a new invariant of the outer protocol:
\begin{tab}
  \S Nq2:/ \> $ \T notify/.q.r = \bot \Land \S copy/.r(q) = \S none/ 
  \Land \T welcome/.q.r \in\{ \bot, \S none/\} $ \\
  \> $ \Implies q\notin\S after/.r \Land \T withdraw/.q.r = 0 $ .
\end{tab}
Preservation of \S Nq2/ at 28, \T hello/, \T after/, and \T notify/
follows from \S Iq1/, \S Iq2/, \S Iq4/, \S Iq5/, \S Kq1/, \S Mq2/, and
the new postulate
\begin{tab}
\S Nq3:/ \> $ \T notify/.q.r \ne \S none/ $ .
\end{tab}
Predicate \S Nq3/ is preserved at 25 because of \S Iq5/. 

Next to \S Lq2/, we need a second invariant about \S prio/:
\begin{tab}
\S Nq4:/ \> $ q\in \S prio/.r \Implies \neg\,\S copy/.r(q) * \S job/.r $ .
\end{tab}
Predicate \S Nq4/ is preserved at 21, 28, and \T after/ because of \S
Lq2/. It is preserved at \T notify/ and \T welcome/ because of \S Iq8/
and \S Kq7/.  This concludes the construction of the invariants for
progress.

At this point, we incorporate the aborting commands of Section
\ref{s.abort}.  These commands preserve all invariants claimed.  In
most cases, they need the same auxiliary invariants as line 28,
because they share several actions with line 28, e.g., resetting \S
pc/, \S job/, \S nbh/, \S nbh0/.

We now summarize the argument for safety by forming the conjunction of
the universal quantifications of the predicates of the families \S
Iq*/, \S Jq*/, \S Kq*/, \S Lq*/, \S Mq*/, \S Nq*/ (the so-called
constituent invariants).  As verified mechanically, this conjunction
is inductive.  Such mechanical verification is relevant, because with
more than 40 invariants, the possibility of overlooking an unjustified
assumption or a clerical error is significant.  As each of the
constituent invariants is a consequence of the conjunction, each of
them is itself invariant, as are all logical consequences of them. In
particular, the mutual exclusion predicate \S Rq0/ is invariant.  This
concludes the proof that the algorithm satisfies \S Rq0/.

One may wonder whether the constituent invariants are independent. We
do believe so, but we have no suitable way to verify this. 

\section{Progress} \label{nodead}

The algorithm satisfies strong progress properties.  In this
section, we introduce and formalize the progress properties.

We introduce weak fairness in Section \ref{intro-wf}.  Section
\ref{execform} contains the formalization in (linear-time) temporal
logic.  Section \ref{inf_live} formalizes weak fairness.  
In Section \ref {s.liveness}, we introduce absence of localized
starvation to unify starvation freedom and concurrency as announced in
Section \ref{progress}, and we announce the main progress result: 
Theorem \ref{liveness}. 

As a preparation for the proof of this result, we derive in Section
\ref{conflicts} some invariants that relate disabledness to the
occurrence of conflicts.  These invariants are used in Section
\ref{nolocdead} to prove absence of localized deadlock.  This result
could be proved as a simple consequence of the main theorem.  We prove
it independently, because it nicely represents the main ideas of the
proof of Theorem \ref{liveness} in a simpler context.

\subsection{Weak fairness} \label{intro-wf}

First, however, weak fairness needs an explanation.  Roughly speaking,
a system is called weakly fair if, whenever some process from some
time onward always can do a step, it will do the step.  Yet if a
process is idle, it must not be forced to be interested in \S
CS/. Similarly, if a process is waiting a long time in the entry
protocol, we do not want it to be forced to abort the entry
protocol. We therefore do not enforce the environment to do steps.  We
thus exclude the environment from the weak fairness conditions.

Formally, we do not argue about the fairness of systems, but
characterize the executions they can perform. Recall that an
\emph{execution} is an infinite sequence of states that starts in an
initial state and for which every pair of subsequent states satisfies
the next-state relation.  The next-state relation is defined as the
union of a number of step relations.  

An execution is called \emph{weakly fair} for a step relation iff,
when the step relation is from some state onward always enabled,
it will eventually be taken.  For example, if some process $p$ is at
line 22, we expect that $p$ will eventually execute line 22.  If some
message $m$ from $q$ to $r$ is in transit, we expect that $r$
eventually receives message $m$.  By imposing weak fairness for some
step relations, we restrict the attention to the executions that are
weakly fair for these step relations.

We partition the steps of our algorithm in four classes (compare the
end of Section \ref{s.abort}).  Firstly, we have the 5
\emph{environment steps} 21, 31, \T ab24/, \T ab25/, \T
ab26/. Secondly, we have the 7 \emph{forward steps} of the lines 22,
23, 24, 25, 26, 27, 28. The third class consists of the 12
\emph{triggered steps}, the ten message reception steps of the
processes and the sites and the two ``delayed answers'': \T after/ and
\T prom/.  The fourth class consists of the two \emph{lowering
  steps} at the lines 32 and 33.

We distinguish the five environment steps because we don't want to
grant them weak fairness.  A process that is idle for a long time must
not be forced to get a job.  A process that stays in its entry
protocol must not be forced to abort it.

\subsection{Formalization in temporal logic} \label{execform}

Let $X$ be the state space of the algorithm.  This is the Cartesian
product of the private state spaces of the processes and the sites,
together with the sets of messages in transit.  Executions are
infinite sequences of states, i.e., elements of the set $X^\omega$.
For a state sequence $\S xs/\in X^\omega$, we write $\S xs/(n)$ for
the $n$th element of \S xs/.  Occasionally, we refer to $\S xs/(n)$ as
the state at time $n$. For a programming variable \T v/, we write $\S
xs/(n).\T v/$ for the value of \T v/ in state $\S xs/(n)$.

For a set of states $U\subseteq X$, we define $\sem{U}\subseteq
X^\omega $ as the set of infinite sequences \S xs/ with $\S xs/(0)\in
U$. For a relation $A\subseteq X^2$, we define $\sem{A}_2\subseteq
X^\omega$ as the set of sequences \S xs/ with $(\S xs/(0),\S
xs/(1))\in A$.

For $\S xs/\in X^\omega$ and $k\in \Nat$, we define the 
sequence $\S drop/(k, \S xs/)$ by $\S drop/(k, \S xs/)(n) = \S xs/(k+n)$.  For a
subset $P\subseteq X^\omega$ we define $\Box P$ (\emph{always} $P$)
and $\Diamond P$ (\emph{eventually} $P$) as the subsets of $X^\omega$
given by
\begin{tab}
\> $\S xs/\in \Box P \EQ (\all k\in\Nat: \S drop/(k, \S xs/)\in P) $  ,\\
\> $\S xs/\in \Diamond P \EQ (\ex k\in\Nat: \S drop/(k, \S xs/)\in P) $  .
\end{tab}

We now apply this to the algorithm.  We write $\S init/\subseteq X$
for the set of initial states and $\S step/\subseteq X^2$ for the next
state relation. Following \cite{AbL91}, we use the convention that
relation \S step/ is reflexive (contains the identity relation).  An
\emph{execution} is an infinite sequence of states that starts in an
initial state and in which each subsequent pair of states is connected
by a step. The set of executions of the algorithm is therefore
\begin{tab}
\> $ \S Ex/ = \sem{\S init/}\cap \Box\sem{\S step/}_2$ .
\end{tab}

If $J$ is an invariant of the system, it holds in all states of every
execution. We therefore have $\S Ex/\subseteq \Box{J}$.

We define $(q\B\ at /\ell)$ to be the subset of $X$ of the states in
which process $q$ is at line $\ell$. An execution in which process $q$
is always eventually at line 21, is therefore an element of
$\Box\Diamond\sem{q\B\ at /21}$.  

\begin{remark} Note the difference between $\Box\Diamond\sem{U}$ and
  $\Diamond\Box\sem{U}$.  In general, $\Box\Diamond\sem{U}$ is a
  bigger set (a weaker condition) than $\Diamond\Box\sem{U}$. The
  first set contains all sequences that are infinitely often in $U$,
  the second set contains the sequences that are from some time onward
  eternally in $U$.  
\end{remark}

\subsection{Weak fairness formalized} 
\label{inf_live}

For a relation $R\subseteq X^2$, we define the \emph{disabled} set
$D(R)=\{x\mid \all y: (x,y)\notin R\}$. Now \emph{weak fairness}
\cite{Lam94} for $R$ is defined as the set of executions in which
$R$ is infinitely often  disabled or taken:
\begin{tab}
\> $ \S wf/(R) \IS \S Ex/\cap 
\Box\Diamond (\sem{D(R)} \cup \sem{R}_2) $ .
\end{tab}

For our algorithm, the next state relation $\S step/\subseteq X^2$ is
the union of the identity relation on $X$ (because \S step/ should be
reflexive) with the relations $\S step/(p)$ that consists of the state
pairs $(x,y)$ where $y$ is a state obtained when process $p$ does a
step starting in $x$.  In accordance with Section \ref{intro-wf}, the
steps that process $p$ can do are partioned as:
\begin{tab}
\> $\S step/(p) \IS \S env/(p)\cup \S fwd/(p)\cup \S trig/(p) 
\cup\S low/(p) $ ,
\end{tab}
where $\S env/(p)$, $\S fwd/(p)$, $\S trig/(p)$, $\S low/(p)$ consist
of the environment steps, the forward steps, the triggered steps, and
the lowering steps of $p$, respectively.  The set of triggered steps
of $p$ is a union:
\begin{tab}
\> $ \S trig/(p) = \bigcup_{q,m} \S rec/(m, q, p) \cup 
\bigcup_{s, n} \S sit/(n,p,s) $ , 
\end{tab}
where $\S rec/(m,q,p)$ consists of the steps where $p$ receives
message $m$ from $q$. Note that we take the union here over all
processes $q$ and the eight message alternatives $m$, including the
delayed answers \T after/ and \T prom/, and $\S sit/(n,p,s)$ consists
of the four commands for message keys $n$ between process $p$ and site
$s$.

The set $\S wf/(\S fwd/(p))$ consists of the executions for which
every forward step of process $p$ is infinitely often disabled or
taken.

The set $\S wf/(\S rec/(m,q,p))$ consists of the executions for which
every message $m$ in transit from $q$ to $p$ is eventually received.

An execution is defined to be \emph{weakly fair for} process $p$ if it
is weakly fair for the forward steps of $p$, for the lowering steps of
$p$, and for all messages with $p$ as destination or source, as
captured in the definition
\begin{tab}
  \> $ \S Wf/(p) = \S wf/(\S fwd/(p))\cap \S wf/(\S low/(p))\cap 
  \bigcap_{s,n} \S wf/ (\S sit/(n,p,s) $\\
  \>\>\> $ \cap \bigcap_{q, m}(\S wf/(\S rec/(m,q,p))\cap \S wf/(\S
  rec/(m,p,q))) $ ,
\end{tab}
where $s$ ranges over the sites, $n$ over the messages to and from
sites, $q$ over the processes, and $m$ over the eight message
alternatives.

\subsection{Absence of localized starvation} \label{s.liveness}

As discussed in Section \ref{progress}, there are two progress
properties to consider: starvation freedom, which means that every
process that needs to enter \S CS/ will eventually do so unless its
entry protocol is aborted, and concurrency, which means that every
process that needs to enter \S CS/ and does not abort its entry
protocol will eventually enter \S CS/ unless it comes in eternal
conflict with some other process.

For starvation freedom we assume weak fairness for the forward and
lowering steps and all triggered steps of all processes.  For
concurrency for process $p$, we only need weak fairness for the
forward and lowering steps of process $p$ itself and for all triggered
steps in which process $p$ is involved.  As we want to verify both
properties with a single proof, we unify them in the concept of
{absence of localized starvation}.  The starting point for the
unification is that the concepts differ in the sets of steps that
satisfy weak fairness.

Let $W$ be a nonempty set of processes.  \emph{Absence of
  $W$-starvation} is defined to mean that weak fairness for all
forward and lowering steps of the processes in $W$ and for all
triggered steps that involve processes in $W$ implies that every
process in $W$ eventually comes back to line 21, unless it comes in
eternal conflict with some process outside $W$.  The special case that
$W$ is the set of all processes is starvation freedom.  The special
case that $W$ is a singleton set is concurrency \cite{ChM84,Rhe98}.

We speak of \emph{absence of localized starvation} if absence of
$W$-starvation holds for every nonempty set $W$.  The aim is thus to
prove that the algorithm satisfies absence of localized starvation. In
order to do so, we formalize the definition in terms of temporal
logic.

We define the set of states where
processes $q$ and $r$ are in conflict as $q\bowtie r$ by
\begin{tab}
\> $ q\bowtie r \EQ \neg\,\S job/.q * \S job/.r$ .
\end{tab}
An execution where $q$ is eventually in eternal conflict with $r$ is
therefore an element of $\Diamond \Box\sem{q \bowtie r}$.  Absence of
$W$-starvation thus means that all ``sufficiently fair'' executions
are elements of the set
\begin{tab}
  \> $ (\bigcap_{q\in W} \Box\Diamond\sem{q\B\ at /21})\: \cup\:
  (\bigcup_{q\in W, r\notin W} \Diamond \Box\sem{q \bowtie r})$ .
\end{tab}

An execution is defined to be weakly fair for a nonempty set of
processes $W$ if it is weakly fair for each of them:
\begin{tab}
  \> $ \S Wf/(W) = \bigcap_{p\in W} \S Wf/(p) $ . 
\end{tab}
Absence of localized starvation thus is the following result:

\begin{theorem} \label{liveness} \ $\S Wf/(W) \subseteq (\bigcap_{q\in
    W} \Box\Diamond\sem{q\B\ at /21}) \cup (\bigcup_{q\in W, r\notin
    W} \Diamond \Box\sem{q \bowtie r})$ holds for every nonempty set
  $W$ of processes.
\end{theorem}

The proof of the Theorem is given in Section \ref{proofThm1}.

\subsection {Disabledness and conflicts} 
\label{conflicts}

As a preparation of the proof of Theorem \ref{liveness}, we use the
invariants obtained in Section \ref{algorithm} to derive four
so-called \emph{waiting invariants} that focus on disabledness of
processes in relation to the occurrence of conflicts.  Forward steps
can be disabled at the lines 23, 24, 25, 26 by nonemptiness of \S
curlist/, \S wack/, \S prio/, \S need/, respectively.  Message
reception is disabled when there is no message.

Let $\S dAfter/ (q,r)$ and $\S dProm/(q,r)$ be the conditions,
respectively, that the alternatives \T after/ and \T prom/ for sending
\T ack/ or \T gra/ from $q$ to $r$ are disabled:
\begin{tab}
\> $ \S dAfter/(q, r) \EQ r\notin\S after/.q\Lor 
\S copy/.q(r) = \S none/ $ ,\\
\> $ \S dProm/(q, r)\EQ r\notin\S prom/.q $\\
\>\>\> $ \Lor\,(q \B\ in /\{27\dots\}\Land 
\neg\,\S job/.q * \S copy/.q(r) ) $ .
\end{tab}

For emptiness of $\S wack/.q$, the invariants \S Iq2/ and \S Nq2/
imply the {waiting invariant}:
\begin{tab}
\S Waq0:/ \> $ \T withdraw/.q.r = \T ack/.r.q = 0\Land 
\T notify/.q.r =  \bot $\\
\> $ \Land\: \T welcome/.q.r \in\{\bot, \S none/\} \Land
\S dAfter/(r, q) \Implies r\notin\S wack/.q $ .
\end{tab}

For emptiness of $\S prio/.q$, the invariants \S Kq7/, \S Lq2/, \S
Mq2/, and \S Nq4/, together with \S Iq4/, \S Iq5/, \S Iq7/, and \S
Iq8/ imply the waiting invariant
\begin{tab}
\S Waq1:/ \> $ r\in \S prio/.q \Land \T withdraw/.r.q = 0 
\Land \T welcome/.q.r \in\{ \bot, \S none/\} $ \\
\> $ \Implies r \B\ in / \{26 \dots\}\Land q\bowtie r$ .
\end{tab}

For emptiness of $\S need/.q$, we are forced to make a case
distinction.  Using the invariants \S Nq0/ and \S Nq1/ together with
\S Iq1/, \S Jq0/, and \S Jq7/, we obtain the waiting invariant
\begin{tab}
\S Waq2:/ \> $ r < q \Land r\in \S need/.q \Land \T withdraw/.r.q = 0 
\Implies r\B\ in / \{26\dots\}\Land q\bowtie r$ .
\end{tab}

It follows from \S Iq4/, \S Iq5/, \S Jq0/, and \S Jq2/, that we have
\begin{tab}
\S Waq3:/ \> $ q < r \Land r\in\S need/.q\Land \T gra/.r.q = 0 
\Land \T notify/.q.r = \bot $ \\
\> $ \Land \S dProm/(r, q) \Implies r \B\ in /\{27\dots\}\Land q \bowtie r$ .
\end{tab}

\subsection{Intermezzo: absence of localized deadlock} \label{nolocdead}

In this section we prove absence of localized deadlock.  This result
is not useful for the proof of Theorem \ref{liveness}, and it follows
from Theorem \ref{liveness}.  Yet, an independent proof of the result
is a good preparation of the more complicated proof that follows.

Informally speaking, absence of localized deadlock means that, when
none of the processes of some set $W$ of processes can do a forward or
lowering or triggered step, and some of them are not at line 21, then
at least one of them is in conflict with a process not in $W$. It is
thus a safety property.

The concept is defined as follows. We define a process $p$ to be
\emph{silent} in some state when every forward or lowering or
triggered step of $p$ is disabled, and no process or site can do a
triggered step that sends a message to $p$.  We define $p$ to be
\emph{locked} when it is silent and not at line 21.

Let $W$ be a set of processes (willing to do steps).  The set $W$ is
said to be \emph{silent} if all its processes are silent.  It is said
to be \emph{locked} if it is silent and contains locked processes.
Absence of $W$-deadlock is the assertion that, if $W$ is locked, then
it contains some process that is in conflict with a process not in
$W$.  Absence of localized deadlock is absence of $W$-deadlock for
every set $W$.  This is our next result:

\begin{theorem} \label{thm-no-deadlock} 
  Assume that a set $W$ of processes is locked. Then there are
  processes $q\in W$ and $r\notin W$ with $q \bowtie r$.
\end{theorem}

\begin{proof}
  The algorithm clearly satisfies the invariant that every process is
  always $\{21\dots 28\}$. The processes in $\{27, 28\}$ can do a
  forward step. Therefore, all processes in $W$ are in $\{21\dots
  26\}$. As $W$ contains locked processes, there are processes in $W$
  at 22--26, waiting for $\S pcr/\leq 32$ or emptiness of \S curlist/,
  \S wack/, \S prio/, \S need/, respectively.  If $p\in W$ has $\S
  pc/.p =22$ then $\S pcr/.p > 32$ and hence $\S pcr/.p = 33$;
  therefore process $p$ can do a lowering step.  This implies that $W$
  has no processes at line 22.

  As all triggered steps for processes in $W$ are disabled, we have
  $\T asklist/.p.s = \T answer/.s.p = \bot $ for all $p\in
  W$. By \S Kq0/, this implies $\S curlist/.p=\emptyset$ for all $p\in
  W$.  It follows that $W$ has no processes at line 23. 

  Similarly, for all pairs $q$, $r$ with $q\in W$ or $r\in W$, the
  values of $\T notify/.q.r$ and $\T welcome/.q.r$ are $\bot$ and the
  values of $\T hello/.q.r $, $ \T withdraw/.q.r$, $\T ack/.q.r$, and
  $\T gra/.q.r$ are all 0.  Also, $\S dAfter/(q,r)$ and $\S dProm/(q,
  r)$ hold.  This simplifies the waiting invariants \S Waq*/ of
  Section \ref{conflicts} considerably. In fact, for $q\in W$ and $r$
  arbitrary, we obtain:
\begin{tab}
\S Wax0:/ \> $ r\notin\S wack/.q $ ,\\
\S Wax1:/ \> $ r\in\S prio/.q\Implies r\B\ in /\{26\dots\}
\Land q\bowtie r $ ,\\
\S Wax2:/ \> $ r < q \Land r\in\S need/.q\Implies r\B\ in /\{26\dots\}
\Land q\bowtie r $ ,\\
\S Wax3:/ \> $ q < r \Land r\in\S need/.q\Implies r\B\ in /\{27\dots\}
\Land q\bowtie r $ .
\end{tab}
It follows from \S Kq1/ and \S Wax0/ that $W$ has no processes
disabled at 24.

Assume that $W$ contains processes disabled at 26. Let $q$ be the
lowest process in $W$ waiting at 26.  Because $q$ is disabled at 26,
the set $\S need/.q$ is nonempty, say $r\in\S need/.q$. It follows
from \S Iq0/ and \S Jq0/ that $r\ne q$. Then \S Wax2/ and \S Wax3/
imply $q\bowtie r$.  Moreover, if $q < r$, then $r$ is in $\{27\dots\}$
by \S Wax3/, so that $r\notin W$. On the other hand, if $r < q$, then
$r$ is in $\{26\dots\}$ by \S Wax2/. If $r\in W$, then $r$ would be at
26, contradicting the minimality of $q$. This proves that $r\notin W$
in either case.

Assume that $W$ contains no processes disabled at 26. Then it has some
process $q\in W$ disabled at 25. Because $q$ is disabled at 25, the
set $\S prio/.q$ is nonempty, say $r\in\S prio/.q$. Now \S Wax1/
implies that $r$ is in $\{26\dots\}$ and $q\bowtie r$. Because $r$ is
in $\{26\dots\}$, it is not in $W$. 
\end{proof}

\section{Verification of Progress} \label{proofThm1}

In this section, we prove progress of the algorithm.  We prepare the
proof of Theorem \ref {liveness} by concentrating on what can be
inferred from weak fairness for a single process.

Following \cite{ChM88}, for sets of states $U$ and $V$, we define
$U\B\ unless /V$ to mean that every step of the algorithm with
precondition $U\setminus V$ has the postcondition $U\cup V$.  The
algorithm has the obvious \B unless/ relations:
\begin{tabn} \label{unlessCen}
\>\+ $ (p\B\ at /\ell)\B\ unless /(p\B\ at /\{\ell + 1 \})$ 
\ for $\ell\in\{22, 23, 27\}$ ,\\
$ (p\B\ at /\ell)\B\ unless /(p\B\ in /\{21, \ell + 1 \})$ 
\ for $\ell\in\{24, 25, 26\} $ ,\\
$ (p\B\ at /28)\B\ unless /(p\B\ at /21)$ .
\end{tabn}
The jumps from 24, 25, 26 back to 21 are due to the aborting commands
of Section \ref{s.abort}.  In any case, all cycles go through line 21:

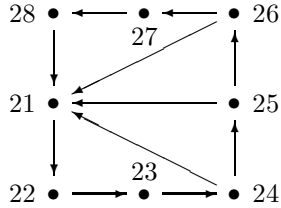
\setlength{\unitlength}{1.2mm} 

\begin{picture}(120,27)(0,2) 
  \multiput(10,0)(0,20) {2}{\multiput(4.2,4.2)(10,0){3} {\bol}}
  \multiput(14.2,14.2)(20,0){2} {\bol}
  \multiput(17,5)(10,0) {2}{\vector(1,0){6}}
  \multiput(23,25)(10,0) {2}{\vector(-1,0){6}}
  \multiput(15,23)(0,-10) {2}{\vector(0,-1){6}}
  \multiput(35,7)(0,10) {2}{\vector(0,1){6}}
  \put(33,6){\vector(-2,1){16}}
  \put(33,15){\vector(-1,0){16}}
  \put(33,24){\vector(-2,-1){16}}
  \put(10,14){21}
  \put(10,4){22}
  \put(10,24){28}
  \put(37,14){25}
  \put(37,4){24}
  \put(37,24){26}
  \put(23.5,6.5){23}
  \put(23.5,21.5){27}
\end{picture}

Weak fairness of process $p$ itself suffices for progress at the lines
27 and 28, because the forward commands at these lines are always
enabled. We thus have
\begin{tabn} \label{notEt22_27_28}
\> $ \S Wf/(p)\cap \Diamond\Box\sem{p\B\ at / \ell} = \emptyset $
\ for $\ell\in\{27, 28\}$ . 
\end{tabn}
We treat progress at the lines 22 up to 26 in the Sections \ref
{progress22} up to \ref {progress26}.  The proof of Theorem \ref
{liveness} is concluded in Section \ref{proofLive}.  Progress of the
lowering thread is proved in Section \ref{progressLow}.

\subsection{Progress at line 22} \label{progress22}

In order to prove progress at line 22, in view of the guard $\S pcr/
\leq 32$, we first need to prove progress of the lowering thread at
line 33.

Assume that the lowering thread of process $p$ remains eternally at
line 33, say from time $n_0$ onward.  From this time onward, the
finite set $\S reglist/.p$ is only modified by removing elements from
it (by \T done/).  Therefore, $\S reglist/.p$ becomes eventually
constant.  If it becomes eventually empty, the step of line 33 is
eventually always enabled, so that process $p$ will move to line 31,
contradicting the assumption. Therefore, there is a site $s$ that
remains eternally in $\S reglist/.p$.  By the invariant \S Kq2/, it
therefore follows that, from time $n_0$ onward, we have $\T lower/.p.s
\ne \bot$ or $ \T done/.s.p > 0$. By weak fairness, any message $\T
lower/.p.s$ is eventually received, and never sent again.  The same
holds for $\T done/.s.p$.  This is a contradiction, thus proving that
\begin{tabn} \label{notEt33}
\> $ \S Wf/(p)\cap \Diamond\Box\sem{p\B\ at / 33} = \emptyset $ .
\end{tabn}
This kind of argument will be used again.  We refer to it as a
\emph{shrinking argument}.

Now assume that process $p$ remains eternally at line 22. Its lowering
thread, if at line 33, will leave line 33, and never come back again
because of the guard of line 32.  Therefore process $p$ is eventually
always enabled at line 22.  By weak fairness, it will eventually leave
line 22, a contradiction.  This proves 
\begin{tabn} \label{notEt22}
\> $ \S Wf/(p)\cap \Diamond\Box\sem{p\B\ at / 22} = \emptyset $ .
\end{tabn}

\subsection {Progress at line 23} \label{progress23}

For some process $p$, assume that \S xs/ is an execution in
$\Diamond\Box\sem{p\B\ at /23}$, weakly fair for $p$.  By the
shrinking argument used above, the set $\S curlist/.p$ is eventually
constant.  Using the invariants \S Kq0/ and \S Kq3/, and weak fairness
for the messages \T asklist/ and \T answer/, we get that $\S
curlist/.p$ is eventually empty.  Then using weak fairness for its
forward steps, we see that process $p$ eventually leaves line 23,
contradicting the assumption.  This proves that
\begin{tabn} \label{notEt23}
\> $ \S Wf/(p)\cap \Diamond\Box\sem{p\B\ at / 23} = \emptyset $ .
\end{tabn}

\subsection{Progress at line 24} \label{progress24}

For some process $p$, assume that \S xs/ is an execution in
$\Diamond\Box\sem{p\B\ at /24}$, weakly fair for $p$.  From some time
$n_0$ onward, process $p$ is and remains at line 24. By the shrinking
argument, the sets $\S pack/.p$ and $\S wack/.p$ are eventually
constant.  By the invariant \S Kq1/ and weak fairness of the messages
\T hello/ and \T welcome/, the set $\S pack/.p$ is eventually empty.
Using weak fairness of the forward steps of $p$, we get that $\S
wack/.p$ remains nonempty.  This implies that there is a process $q$
such that, from time $n_0$ onward, $q\in \S wack/.p$ always holds.

We next note that, as process $p$ is eventually always at line 24, we
have eventually always $\T notify/.p.q = \bot$ and $\T withdraw/.p.q =
0$ and $\T welcome/.p.q\in\{\bot,\S none/\}$.  In fact, once process
$p$ is and remains at line 24, it will not send any messages \T
notify/ or \T withdraw/, or any messages $\T welcome/ \ne \S none/$.
Weak fairness ensures that any such messages are received eventually.
This proves that there is a time $n_1\geq n_0$ such that, from time
$n_1$ onward, we have $\T notify/.p.q = \bot$ and $\T withdraw/.p.q =
0$ and $\T welcome/.p.q\in\{\bot,\S none/\}$.  Now \S Waq0/ implies
that $\neg\,\S dAfter/(p, q)\lor \T ack/.q.p > 0$ holds from time
$n_1$ onward. By weak fairness, process $q$ will send \T ack/ to $p$.
By weak fairness, this \T ack/ will be received, and $p$ will remove
$q$ from $\S wack/.p$.  This contradiction proves:
\begin{tabn} \label{notEt24}
\> $ \S Wf/(p) \cap \Diamond\Box\sem{p\B\ at /24} = \emptyset$ .
\end{tabn}

\subsection{Progress at line 25} \label{progress25}

For some process $p$, assume that \S xs/ is an execution in
$\Diamond\Box\sem{p\B\ at /25}$, weakly fair for $p$.  From some time
$n_0$ onward, process $p$ is and remains at line 25. By the shrinking
argument, there is a process $q$ such that, from time $n_0$ onward,
$q\in \S prio/.p$ always holds.  If $\T withdraw/.q.p > 0$ holds at
some time $n\geq n_0$, by weak fairness this message will eventually
arrive and falsify $q\in \S prio/.p$. Therefore, $\T withdraw/.q.p =
0$ holds from time $n_0$ onward.  By the above argument, there is a
time $n_1\geq n_0$ such that $\T welcome/.p.q\in\{\bot,\S none/\}$
holds from time $n_1$ onward.  By \S Waq1/, $q$ is in $\{26\dots\}$
and $p\bowtie q$ holds from time $n_1$ onward. This proves
\begin{tabn} \label{notEt25}
\> $ \S Wf/(p) \cap \Diamond\Box\sem{p\B\ at /25} \subseteq
\bigcup_q \Diamond\Box\sem{q\B\ in /\{26\dots\} \Land p\bowtie q }$ .
\end{tabn}

\subsection{Progress at line 26} \label{progress26}

For some process $p$, let \S xs/ be an execution in $\Diamond \Box
\sem {p\B\ at /26}$, weakly fair for $p$.  Process $p$ waits at line
26 for emptiness of \S need/. This condition belongs to the inner
protocol.  The inner protocol in isolation, however, is not starvation
free because it would allow a lower process repeatedly to claim
priority over $p$ by sending notifications.  We need the outer
protocol to preclude this.  Technically, the problem is that $\S
need/.p$ can grow at line 26 because of the alternative \T prom/.

We therefore investigate conditions under which the predicate $q\in\S
need/.p$ or its negation is stable, i.e., once true remains true.  As
$p$ remains at line 26, the set $\S nbh0/.p$ remains constant.  For
$q\notin\S nbh0/.p$, we have $q\notin \S need/.p$ by \S Jq0/.  For
$q\in\S nbh0/.p$, we define the predicate
\begin{tab}
\> $ \S bf/(q, p): \quad q\B\ in /\{\dots 24\} \Lor
 \S job/.p * \S job/.q \Lor p\in\S prio/.q $ .
\end{tab} 
Roughly speaking, $ \S bf/(q, p)$ expresses that $q$ is not in the
inner protocol or is not in conflict with $p$.  While process $p$ is
and remains at line 26 and $\S copy/.q(p)\ne\S none/$, the predicate
$\S bf/(q, p)$ is stable. The main point is when process $q$ executes
line 24.  If the second disjunct of $\S bf/(q,p)$ does not
hold, process $q$ puts $p$ into $\S prio/.q$. The proof uses \S Iq2a/,
\S Iq4/, \S Iq7/, \S Iq8/, and \S Rq1a/.

While process $p$ is and remains at line 26 and either $\S bf/(q,p)$
holds or $p < q$, the predicate $q\notin\S need/.p$ is stable. This is
proved at the alternative \T prom/ with \S Iq4/, \S Jq0/, \S Jq1/, \S
Jq2/, and \S Lq2/.

While process $p$ is and remains at line 26 and $\S bf/(q,p)$ is false
and $q\leq p$, the predicate $q\in\S need/.p$ is stable. This is
proved at the alternatives \T gra/ and \T withdraw/ with \S Iq2a/, \S
Jq5/.

We can now combine these predicates in the variant function
\begin{tab}
  \> $ \S vf/(q, p) \IS (\S bf/(q, p)\: ? \quad |\,q\in\S need/.p\,| $ \\
\>\>\>\>\> $: \; q\in\S need/.p\: ?\quad 3 $\\
\>\>\>\>\> $: \; p < q\: ?\quad 2 :\; 4) $ .
\end{tab}
Here, we write $(B\,?\:x:y)$ for the conditional expression that
equals $x$ when $B$ holds, and otherwise $y$, as in the programming
language C.  Clearly, $\S vf/(q, p)$ is odd (1 or 3) if and only if
$q\in\S need/.p$ holds. Similarly, $\S vf/(q, p) \leq 1$ holds if and
only if $\S bf/(q, p)$.

The above stability results about $\S bf/(q, p)$ and $q\in\S need/.p$
and its negation imply that, while $p$ is and remains at line 26 and
$\S copy/.q(p)\ne\S none/$, the function $\S vf/(q, p)$ never
increases.  By weak fairness, eventually $\T notify/.p.q = \bot$
holds.  As process $p$ remains at line 26, $\T notify/.p.q = \bot $
remains valid.  The invariant \S Iq4/ together with \S Iq5/ and \S
Iq1/ therefore implies that $\S copy/.q.(p)\ne\S none/$ holds and
remains valid because $q\in\S nbh0/.p$.  Therefore, from that time
onward, $\S vf/(q, p)$ never increases.

It follows that, for every $q\in\S nbh0/.p$, eventually, $\S vf/(q, p)$
gets a constant value.  Therefore, the truth value of $q\in\S need/.p$
is also eventually constant.  Finally, as $\S need/.p$ is always a
subset of the finite set $\S nbh0/.p$, which is constant while $p$ is
at line 26, we can now conclude that the set $\S need/.p$ is
eventually constant.

If the set $\S need/.p$ is eventually empty, process $p$ would be
eventually always enabled. Weak fairness of $p$ would then imply that
process $p$ would leave line 26. Therefore, there is some process $q$
eventually always in $\S need/.p$. We have $q\ne p$ because of \S Jq0/
and \S Iq0/. This proves
\begin{tabn} \label{evtNeed} 
  \> $ \S Wf/(p) \cap \Diamond\Box\sem{p\B\ at /26} 
  \subseteq \bigcup_{q\ne p} \Diamond\Box \sem{q\in\S need/.p} $ .
\end{tabn}

Assume that $q\in\S need/.p$ holds from time $n_0$ onward.  We now
make a case distinction.  First assume $q < p$.  If $\T withdraw/.q.p
> 0$ holds at some time $n\geq n_0$, weak fairness implies that the
message \T withdraw/ will be received at some time $ > n_0$, which
would falsify $q\in\S need/.p$. Therefore, $\T withdraw/.q.p = 0$
holds from time $n_0$ onward.  By predicate \S Waq2/, we thus have
$q\B\ in /\{26\dots\}$ and $p\bowtie q$ from time $n_0$ onward.  This
proves
\begin{tabn} \label{evtLower} 
  \> $ q < p \Implies \S Wf/(p) \cap 
  \Diamond\Box\sem{q\in \S need/.p}
  \subseteq  \Diamond\Box \sem{q\B\ in /\{26\dots\}\land p\bowtie q} $ .
\end{tabn}

Next, assume $p < q$.  If $\T gra/.q.p > 0$ holds at some time $n\geq
n_0$, this \T gra/ message will be received because of weak fairness,
falsifying $q\in\S need/.p$. Therefore, $\T gra/.q.p = 0$ holds from
time $n_0$ onward. Also by weak fairness, we have $\T notify/.p.q =
\bot $ at some time $n_1 \geq n_0$. Because process $p$ is and remains
at line 26, it cannot send such messages again. Therefore, $\T
notify/.p.q = \bot $ holds from time $n_1$ onward.  Because process
$q$ sends no \T gra/ message to $p$ after $n_0$, weak fairness implies
that $\S dProm/(q, p)$ holds infinitely often after $n_1$.  By \S
Waq3/ this implies that, infinitely often, we have
\begin{tab}
\>  $ \S bg/(q, p): \quad q\B\ in /\{27\dots\} \Land p\in\S nbh/.q 
\Land p\bowtie q $ .
\end{tab}
We need the condition $\S bg/(q, p)$ eternally, however, not only
infinitely often.  For this purpose, we reuse that $\S vf/(q,p)$ is
eventually constant, say that $\S vf/(q,p) = k$ holds from time
$n_2\geq n_1$ onward. As condition $\S bg/(q,p)$ contradicts $\S
bf/(q,p)$, we have $k >1$.  Therefore, $\S bf/(q,p)$ is false from
time $n_2$ onward. This implies that process $q$ does not execute line
28 from time $n_2$ onward. Therefore, condition $\S bg/(q,p)$ holds
eventually always.  Consequently, we obtain
\begin{tabn} \label{evtHigher} 
  \> $ p < q \Implies \S Wf/(p) \cap \Diamond\Box\sem{q\in \S need/.p}
  \subseteq  \Diamond\Box \sem{q\B\ in /\{27\dots\}\land p\bowtie q} $ .
\end{tabn}

\subsection{The proof of Theorem \ref{liveness}} \label{proofLive}

Let $W$ be a nonempty set of processes. Let us introduce the abbreviation
\begin{tab}
\> $\S cf/(W)=\bigcup_{q\in W,r\notin W} \Diamond\Box\sem{ q\bowtie r} $ .
\end{tab} 
Now Theorem \ref{liveness} asserts that $\S Wf/(W)\subseteq
\Diamond\Box\sem{p\B\ at /21}\cup \S cf/(W)$ for every $p\in W$.

We first combine the results of Section \ref{progress26} to prove
\begin{tabn} \label{at26} 
  \> $ p\in W \Implies\S Wf/(W) \cap \Diamond\Box\sem{p\B\ at /26}
  \subseteq \S cf/(W) $ .
\end{tabn}
Let \S xs/ be an execution in the lefthand set. Because $W$ contains
some process that remains eternally at 26, we can consider the lowest
process, say $q\in W$, with this property. By formula (\ref{evtNeed}),
there is a process $r\ne q$ that remains eternally in $\S need/.q$. If
$q < r$, formula (\ref{evtHigher}) implies that $r$ remains eternally
in $\{27\dots\}$ and in conflict with $q$.  By Formula
(\ref{notEt22_27_28}), it follows that $r\notin W$ and hence that $\S
xs/ \in \S cf/(W)$. On the other hand, if $r < q$, Formula
(\ref{evtLower}) implies that $r$ remains eternally in $\{26\dots\}$
and in conflict with $q$.  If $r\in W$, minimality of $q$ implies that
process $r$ proceeds to $\{27,28\}$, contradicting Formula
(\ref{notEt22_27_28}).  Therefore $r\notin W$, and $\S xs/\in\S
cf/(W)$. This concludes the proof of Formula (\ref{at26}).

By similar arguments, we use Formula (\ref{notEt25}) to obtain:
\begin{tabn} \label{at25} 
  \> $ p\in W \Implies\S Wf/(W) \cap \Diamond\Box\sem{p\B\ at /25}
  \subseteq \S cf/(W) $ .
\end{tabn}

Finally, Theorem \ref{liveness} follows from the Formulas (\ref{unlessCen}),
(\ref{notEt22_27_28}), (\ref{notEt23}), (\ref{notEt24}), 
(\ref{at26}), and (\ref{at25}).  

\subsection{Progress for lowering} \label{progressLow}

Formula (\ref{notEt33}) gives progress of the lowering thread at line 33,
under the assumption of weak fairness.  Progress at line 32, however,
requires more than weak fairness.  While the lowering thread of $p$ is
at line 32, the main thread of $p$ can repeatedly enter and leave the
region 22---24.  Therefore, the step of line 32 is not eventually
always enabled.  Indeed, we do not want that resource acquisition
suffers for lowering. 

We need strong fairness for progress at line 32.  Strong fairness at
line 32 means that, if process $p$ is infinitely often enabled at line
32, it will eventually take the step of line 32.  Formally, for a
relation $R\subseteq X^2$, \emph{strong fairness} \cite{Lam94} for $R$
is defined as the set of executions in which $R$ is eventually always
disabled or infinitely often taken:
\begin{tab}
  \> $ \S sf/(R) \IS \S Ex/\cap ( \Diamond\Box \sem{D(R)} \cup
  \Box\Diamond \sem{R}_2 ) $ .
\end{tab}

If, in some execution, process $p$ is always eventually at line 21,
and yet remains at line 32, the step of line 32 is infinitely often
enabled, so that indeed strong fairness guarantees that the step will
be taken.  In combination with formula (\ref{notEt33}), we thus obtain
progress for the lowering thread in the sense that it always returns
to its idle state at line 31:
\begin{tab}
\> $ \Box\Diamond\sem{p\B\ at /21} \cap \S Wf/(p)\cap\S sf/(\S st32/(p))
\subseteq \Box\Diamond\sem{p\B\ at /31} $ ,
\end{tab}
where $\S st32/(p)$ is the relation corresponding to the step of $p$
at line 32.


\section {Message Complexity and Waiting Times} 
\label{summ_alg}

In the central algorithm, a process exchanges 3 or 4 messages with
every neighbour.  In the querying phase, at lines 22 and 23, it
exchanges 2 messages with every site it is interested in, plus 2
messages for every potential competitor.

In a message passing algorithm, we can distinguish two kinds of
waiting.  There is waiting for answers that can be sent immediately.
Such waiting requires at most $2\Delta$, because $\Delta$ is an upper
bound of the time needed for the execution of an alternative plus the
time the messages sent are in transit.  In the central algorithm, this
happens at line 23, for emptiness of \S curlist/, and at line 24, for
emptiness of \S pack/ and \S wack/.  Waiting in the lowering
algorithm, at line 33, is also of this kind.

The other kind of waiting is when a process needs to wait for the
progress of other processes.  These are the important waiting
conditions.  The central algorithm has two locations where this is the
case.  At line 25, the process waits for emptiness of \S prio/ to make
accumulation of conflicting processes unlikely.  At line 26, it waits
for emptiness of \S need/ to ensure partial mutual exclusion.

The waiting time $T_1$ for emptiness of \S need/ at line 26 depends on
the conflict graph of the processes that are concurrently in the inner
protocol. The middle layer tries to keep this graph small by
guaranteeing that conflicting processes do not enter the inner
protocol concurrently unless they pass line 25 within a period
$\Delta$.

The waiting time $T_2$ for emptiness of \S prio/ in line 25 depends on
the efficiency of the inner protocol, because for a process $p$ the
elements of $\S prio/.p$ are in the inner protocol and are removed
from $ \S prio/.p$ when they withdraw.  We thus have $T_2\leq
T_1+\Gamma+\Delta$ where $\Gamma$ is an upper bound for the time spent
in \S CS/.  Indeed, every element of \S prio/ arrives in \S CS/ after
time $T_1$, and at line 28 after $T_1+\Gamma$, while the message \T
withdraw/ takes time $\Delta$.  

The total waiting time for the main loop body is at most $6\Delta +
T_2 + T_1 + \Gamma \leq 2T_1 + 2\Gamma + 7\Delta $.  This includes
$2\Delta$ for concurrent lowering.

The value of $T_1$ heavily depends on the load of the system
and other system parameters: for resource $c$, the number of jobs
activated per $\Gamma$ that need resource $c$; the number of resources
per job; the number of resources per site; the number of processes per
site; the number of conflicts per job.  Of course, all these numbers
should be averages, and they are not independent.  Experiments are
needed to evaluate the performance of the algorithm, and to compare it
with other algorithms.

\section{Conclusions} \label{conclusion}

The problem of distributed resource allocation is a matter of partial
mutual exclusion, with the partiality determined by a dynamic conflict
graph.  Our solution allows infinitely many processes, and it allows
conflicts between every pair of processes.  The primary
disentanglement is the split into the central algorithm and the
registration algorithm.

In the central algorithm, the conflict graph is dynamic but limited by
the current registrations, and the jobs can be treated as
uninterpreted objects with a compatibility relation.  The central
algorithm itself consists of three layers.  In the outer protocol, the
processes communicate their jobs.  In the inner protocol they compete
for the critical section.  The middle layer protects the inner
protocol from flooding with conflicting processes.

The neighbourhoods used in the central algorithm are formed in a
querying phase in which the processes communicate with a finite number
of registration sites.  We use a flexible job model that allows e.g.
the distinction between read permissions and write permissions.  We
reach a fully dynamic conflict graph by enabling the processes to
modify their registrations.

Our solution does not automatically satisfy the ``economy'' condition
of \cite{ChM84} that permanently tranquil philosophers should not send
or receive an infinite number of messages. Indeed, in our algorithm, a
permanently idle process that occurs infinitely often in the
neighbourhood of other processes will receive and send infinitely many
messages.  It can avoid this, however, by resetting its registrations
to zero.

Our solution is more concurrent than the layered solution of
\cite{WeL93}. It satisfies the requirement that, ``if a drinker
requests a set $B$ of bottles, it should eventually enter its critical
region, as long as no other drinker uses or wants any of the bottles
in $B$ forever'' (\cite[p.\ 243]{WeL93}).

Our algorithm does not minimize the response time.  Yet, it may
perform reasonably well in this respect, because the middle layer of
the central algorithm prohibits entrance for new processes that have
known conflicts with processes currently in the inner protocol.  In
the inner protocol, the lower processes have the advantage that they
can force priority over higher conflicting processes.  When conflicts
in the inner protocol are rare, however, this bias towards the lower
processes will not be noticeable.

The algorithm as presented allows several simplifications.  (1) The
aborting commands of Section \ref{s.abort} can be removed.  (2) One
can decide to give every resource its own site, or to use a single
site for all resources.  (3) If one takes the simplest job model,
i.e. $K=1$ in Section \ref{jobmodel}, the arrays \S fun/, \S news/,
and \S list/ reduce to finite sets.  (4) One can fix the network
topology, i.e., replace the variables $\S nbh/.p$ by constants, and
remove the registration algorithm.  (5) If the set of all processes is
a finite and rather small set $\S Proc/$, one can even take $\S nbh/.p
= \S Proc/$ for all $p$.

It would be interesting to see how much the algorithm can be
simplified by using reliable synchronous messages, or how the
algorithm can be made robust by allowing the asynchronous messages to
be lost (and possibly duplicated).

The algorithm could not have been designed without a proof assistant
like PVS.  This holds in particular for the proofs of safety of
registration (the invariants \S Mq*/), the proof of progress of the
central algorithm (Section \ref{progress26}), and the use of array \S
copy/ in the registration algorithm to avoid additional waiting.

\sbreak \B Acknowledgement./ The observation that the algorithm also
solves the readers/writers problem was made by Arnold Meijster.

\bibliographystyle{plain} 
\bibliography{refs}

\end{document}